\newif\ifarxiv

\arxivtrue 

\ifarxiv
\documentclass{article}
\usepackage{arxiv}
\usepackage{amsthm}
\else

\documentclass{article}

\usepackage{icml2023}

\fi

\usepackage{microtype}
\usepackage{graphicx}
\usepackage{subfigure}
\usepackage{booktabs} 

\usepackage{hyperref}

\usepackage{amsmath}
\usepackage{amssymb}
\usepackage{mathtools}
\usepackage{amsthm}

\usepackage[capitalize,noabbrev]{cleveref}

\usepackage[utf8]{inputenc} 
\usepackage[T1]{fontenc}    
\usepackage{hyperref}       
\usepackage{url}            
\usepackage{booktabs}       
\usepackage{amsfonts}       
\usepackage{nicefrac}       
\usepackage{microtype}      
\usepackage{xcolor}         

\usepackage{comment} 
\usepackage{lipsum} 
\usepackage{graphicx}
\usepackage{makecell}
\usepackage{bbm}
\usepackage{epsfig, latexsym, amsfonts, amssymb, graphicx}
\usepackage{soul}
\usepackage{tikz} 
\usepackage{enumitem}
\usetikzlibrary{positioning}
\usepackage{comment}
\usepackage[american]{circuitikz}
\usepackage[algo2e, norelsize, linesnumbered, ruled, lined, boxed, commentsnumbered]{algorithm2e}
\usepackage{times, fancyhdr,xcolor, algorithm,  natbib, eso-pic, forloop, url}

\usepackage[noend]{algpseudocode}
\usetikzlibrary{shapes,arrows}

\DeclareMathOperator*{\argmin}{arg\,min}
\newcommand{\R}{\mathbb{R}}
\newcommand{\x}{\mathbf{x}}
\newcommand{\z}{\mathbf{z}}
\newcommand{\ahat}{\hat{\mathbf{a}}}
\newcommand{\fp}{\hat\x} 
\newcommand{\y}{\mathbf{y}}

\newcommand{\dfdx}{\nabla f(\x)}
\newcommand{\dfdy}{\nabla f(\y)}
\newcommand{\dfdxt}{\nabla f(\xt)}
\newcommand{\dxdt}{\dot{\x}}
\newcommand{\xt}{\x(t)}
\newcommand{\Zeta}{Z}
\newcommand{\Ab}{\mathbf{A}}

\usepackage{accents}

\usepackage[export]{adjustbox}
\algdef{SE}[DOWHILE]{Do}{doWhile}{\algorithmicdo}[1]{\algorithmicwhile\ #1}%
\renewcommand{\cite}[1]{\citep{#1}}
\makeatletter
\def\thanks#1{\protected@xdef\@thanks{\@thanks
        \protect\footnotetext{#1}}}
\makeatother

\theoremstyle{plain}
\newtheorem{theorem}{Theorem}[section]
\newtheorem{proposition}[theorem]{Proposition}
\newtheorem{lemma}[theorem]{Lemma}
\newtheorem{corollary}[theorem]{Corollary}
\theoremstyle{definition}
\newtheorem{definition}[theorem]{Definition}
\newtheorem{assumption}[theorem]{Assumption}
\theoremstyle{remark}
\newtheorem{remark}[theorem]{Remark}

\usepackage[textsize=tiny]{todonotes}

\ifarxiv
\else
\icmltitlerunning{Submission and Formatting Instructions for ICML 2023}
\fi

\begin{document}

\ifarxiv
\title{ECCO: Equivalent Circuit Controlled Optimization}

\author{Aayushya Agarwal$^{1*}$, Carmel Fiscko$^{1*}$, Soummya Kar$^{1}$, Larry Pileggi$^{1}$, and Bruno Sinopoli$^{2}$
\thanks{\hspace{-2em}$*$These authors contributed equally.}%
\thanks{\hspace{-2em}$^{1}$Aayushya Agarwal, Carmel Fiscko, Soummya Kar, and Larry Pileggi are with the Dept. of Electrical and Computer Engineering at Carnegie Mellon University at 5000 Forbes Ave, Pittsburgh, PA 15213.}
\thanks{\hspace{-2em}$^{2}$Bruno Sinopoli is with the Dept. of Electrical and Systems Engineering at Washington University in St. Louis, MO at 1 Brookings Dr, St. Louis, MO 63130. {\tt\small bsinopoli@wustl.edu }}%
}
\maketitle

\fi
\begin{abstract}
We propose an adaptive optimization algorithm for solving unconstrained scaled gradient flow problems that achieves fast convergence by controlling the optimization trajectory shape and the discretization step sizes. Under a broad class of scaling functions, we establish convergence of the proposed approach to critical points of smooth objective functions, while demonstrating its flexibility and robustness with respect to hyperparameter tuning. First, we prove convergence of component-wise scaled gradient flow to a critical point under regularity conditions. We show that this controlled gradient flow dynamics is equivalent to the transient response of an electrical circuit, allowing for circuit theory concepts to solve the problem. Based on this equivalence, we develop two optimization trajectory control schemes based on minimizing the charge stored in the circuit: a second order method that uses the true Hessian and an alternate first order method that approximates the optimization trajectory with only gradient information. While the control schemes are derived from circuit concepts, no circuit knowledge is needed to implement the algorithms. To find the value of the critical point, we propose a time step search routine for Forward Euler discretization that controls the local truncation error, a method adapted from circuit simulation ideas. In simulation we find that the trajectory control outperforms uncontrolled gradient flow, and the error-aware discretization out-performs line search with the Armijo condition. Our algorithms are evaluated on convex and non-convex test functions, including neural networks, with convergence speeds comparable to or exceeding Adam. 
\end{abstract}

\section{Introduction}
\label{sec:Intro}
Optimization is a key problem across all areas of science, research, and engineering. Popular first-order methods like gradient descent are commonly explained with the metaphor of "a ball rolling down a hill." The physics analogy provided intuition for extensions like acceleration methods. However, prominent methods such as Adam (\cite{kingma2014adam}) have relied more on the intuition of experienced researchers to propose new ideas than returning to the physics intuition. 

To study physics-inspired optimization techniques, we focus on the paradigm of gradient flow, which models the continuous-time trajectory of an optimization variable. Gradient flow techniques have been of key interest, as they allow convergence properties to be analyzed in terms of the functional trajectory. Gradient flow also allows an optimization process to be described as a dynamical control system.

In this work, we present a new structure for optimization whereby a component-wise scaled gradient flow can be \emph{modeled} as the transient response of an equivalent circuit (EC). The steady-states of this EC represent critical points of the objective function. Solving the optimization problem, therefore, may be achieved within the domain of solving the circuit. This work emphasizes the parallels between optimization and circuits, and demonstrates how circuit domain knowledge can be used to solve optimization problems. 

In this work, we propose using circuits to solve optimization problems for three main reasons. First, circuits may be built to produce general signals; this allows a circuit to model a variable in an objective function as the voltage at a node. The circuit, and therefore the optimization problem, can then leverage tools from circuit physics, engineering, and control literature. This enables concepts like electrical charge and energy to be used for the theoretical analysis and the proposal of better adaptive optimization tools.

The second motivation is that top-of-the-line circuit simulation software is capable of efficiently simulating systems of millions of variables and could solve optimization problems at massive \emph{scale}. Decades of research have produced sophisticated circuit simulation techniques, which combine circuit physics, numerical tricks, and domain knowledge into industry-standard tools. By representing optimization problems as circuits, the optimization community may take advantage of these key, currently-overlooked resources.

Finally, a long-term goal could be to build an analog "physical optimizer" machine based on this work; however, this goal is beyond the scope of the paper as here we solely focus on using circuit analysis to design optimization algorithms.




In this work, we first prove the convergence of general component-wise scaled gradient flow to a critical point of the objective function. Our goal then becomes to design proper scaling factors to accelerate the convergence speed. 

The scaling factors will be designed based on ideas from circuit theory. To this end, we design a general-purpose equivalent circuit such that the circuit's transient response is equal to a solution of the scaled gradient flow ODE. The node voltages of the EC are equal to the optimization variable $\x$ and the capacitor values model the scaling factors of the component-wise scaled gradient flow. The resulting design problem reduces to the selection of the capacitor values as a function of the node voltage to quickly dissipate stored charge. Once the circuit has dissipated all its charge, it has reached steady-state, the ODE has reached a fixed point, and the gradient flow has reached a critical point.

After establishing convergence of the trajectory to a critical point, the problem is ready to be solved numerically for the value of that critical point. While discretization techniques have been widely used in optimization theory, they have also been extensively studied for decades in the paradigm of circuit simulation, and the circuits literature has developed tools overlooked in the optimization world. In particular, circuit simulation chooses step sizes that bound the local truncation error (LTE) produced by an approximation in the numerical integration step. We adapt these ideas into a discretization process amenable for optimization algorithms. 

To the best of our knowledge, this is the first work formally using circuit theory for gradient flow, and using circuit simulation techniques to solve general optimization problems.



\textbf{Our main contributions} in this paper are 1) proving convergence of component-wise scaled gradient flow to a critical point, 2) proposing control policies for scaling factors to accelerate convergence speed, and 3) introducing LTE as a condition for backtracking line search in iterative optimization algorithms. In simulation we show that our controlled gradient flow algorithm converges faster than uncontrolled gradient flow, the LTE critera outperforms the Armijo condition for forward Euler discretization, and our overall algorithm is much less sensitive to hyperparameter tuning than Adam. This work is directly applicable to neural networks. No circuit background knowledge is needed to implement the proposed algorithms.

The problem formulation is presented in \S\ref{sec:Problem}. We show the component-wise scaled gradient flow in \S\ref{sec:gflow}, and formulate the equivalent circuit in \S\ref{sec:ec}. The control scheme is presented in \S\ref{sec:control}, followed by an approximate control scheme that does not use a Hessian in \S\ref{sec:acontrol}. Error-aware discretization techniques are presented in \S\ref{sec:discret}. We propose a specific algorithm that a practitioner can implement directly in \S\ref{sec:Algo}. The paper concludes with simulations including a power systems example and a sensitivity analysis to hyperparameters on a neural network in \S\ref{sec:Sims}.

\section{Literature Review}
\label{sec:Lit}
Gradient flow methods have been well-studied due to their potential to draw general conclusions in the continuous-time domain \cite{behrman1998efficient}, \cite{attouch1996dynamical}, \cite{brown1989some}. The gradient-flow formulation is alluring because it enables theoretical guarantees to be made without any error introduced by a discretization process.
Recent interest in gradient flow methods has given convergence analyses of scaled and normalized gradient flows \cite{murray2019revisiting}, distributed techniques \cite{swenson2021distributed}, momentum \cite{muehlebach2021optimization}  \cite{franca2018dynamical}, stochastic gradient descent \cite{latz2021analysis}, continuous-time mirror descent \cite{amid2020reparameterizing}, and analyzing ADMM \cite{franca2018admm}. 

 Gradient flow can be viewed as a dynamical system, thereby introducing concepts from control theory \cite{helmke2012optimization}, \cite{yuille1994statistical}.
 One important methodology in control systems is provided by Lyapunov theory, widely used to show stability of dynamical systems, which has emerged as a tool to show convergence of gradient flow \cite{cortes2006finite},  \cite{wilson2018lyapunov}, \cite{wilson2021lyapunov}, \cite{polyak2017lyapunov}, \cite{hustig2019robust}. 
 For example, analogies can be made with energy dissipation in physics, yielding insights into Lyapunov function construction and convergence analysis \cite{hu2017dissipativity}. We are only aware of one recent reference that uses circuits to model an optimization problem \cite{boyd2021distributed}, but there are several fundamental differences between this work and ours. \cite{boyd2021distributed} shows a circuit interpretation of distributed optimization, where sub-problems are connected via wires and the gradient of the objective is modeled using a nonlinear resistor. Voltages and currents represent primal variables and dual residuals, and concepts of convexity connect with passivity of the nonlinear resistors to establish convergence. In comparison, our approach shows that the voltage-current equations of a fully-connected circuit are equal to the scaled gradient flow differential equation, and we use an adjoint circuit and capacitor charge dissipation to design a trajectory that will quickly reach steady state. 

While useful theoretical results can be established in continuous-time, solving the ODE with a computer generally necessitates some discretization scheme. Simple examples are using constant or diminishing step sizes, or line search methods such as backtracking line search and the Wolfe conditions. Recent advances have used sophisticated explicit integration techniques to approximate the continuous system \cite{pmlr-v97-muehlebach19a}, \cite{lin2016distributed}, \cite{andrei2004gradient}, \cite{scieur2017integration}. Further discretization methods have been explored including implicit integration methods such as backward-Euler \cite{barrett2020implicit}. Few of these works \cite{barrett2020implicit},\cite{scieur2017integration}, \cite{andrei2004gradient} have used notions of LTE to determine appropriate time-steps for their respective numerical integration method. However, the idea of adapting step sizes based on LTE and stability is not well explored.

In this work, we draw on industrial discretization techniques developed for circuit simulation. One simulator, SPICE \cite{nagel1971computer}, analyzes the transient response of general analog circuits by solving an underlying stiff, nonlinear ODE. Unlike generic solvers, SPICE uses circuit physics to develop intuitive heuristics to solve the ODE.

\section{Problem Formulation}
\label{sec:Problem}
In this work, we consider the following unconstrained optimization problem:
    \begin{gather}
        \min_{\x} f(\x),\label{prob}\\
        \x^* \in \argmin_{\x} f(\x).
    \end{gather}
    where $\x\in \R^n$ and $f: \R^n\to\R$. It is known \cite{brown1989some} that  \eqref{prob} may be solved via the gradient flow ODE initial value problem (IVP),
    \begin{equation}
        \dxdt(t)=-\dfdxt,\quad \x(0)=\x_0, \label{gradient flow}
    \end{equation}
    where $\dxdt(t)$ refers to the time derivative of $\xt$. We consider component-wise scaled gradient flow IVP, where the RHS of \eqref{gradient flow} is multiplied by a positive diagonal matrix $Z^{-1}$.
    \begin{equation}
        \dxdt(t)=-Z(\xt)^{-1} \nabla f(\xt),\quad \x(0) = \x_0. \label{scaled gradient flow}
    \end{equation}
    
    The solution to an ODE IVP is computed with the integral $\x(t) = \x(0) + \int_0^t \dxdt(s) ds.$ In this work we assume that the following set of assumptions are satisfied:

    \begin{enumerate}[wide=\parindent,label=\textbf{(A\arabic*)}]
        \item $f\in C^2$ and $\inf_{\x\in\R^n}f(\x)>-R$ for some $R>0$. \label{a1}
        \item $f$ is coercive, i.e., $\lim_{\|\x\|\to\infty} f(\x) = +\infty$. \label{a2}
        \item $\nabla^2 f(\x)$ is non-degenerate. \label{a3}
        \item (Lipschitz and bounded gradients): for all $\x,\y\in\R^n$,  $\Vert \dfdx-\dfdy\Vert\leq L\Vert\x-\y\Vert$, and $\Vert\dfdx\Vert\leq B$ for some $B>0$. \label{a4}
        \item $Z(\x)^{-1}$ is diagonal for all $\x$ and $d_2>Z_{ii}(\x)^{-1}>d_1$ for all $i,\ \x$ and for some $d_1,\ d_2>0$.\label{z def}
    \end{enumerate}

    
    Note that $f$ is not assumed to be convex. Coercivity guarantees that there exists a finite global minimum for $f(\x)$ that is differentiable \cite{peressini1988mathematics}. 
    
    \begin{definition}
    We say $\x$ is a \emph{critical point} of $f$ if it satisfies $\nabla_{\x}f(\x) = \vec{0}$. Let $S$ be the set of \emph{critical points}, i.e. $S=\{\x\ |\ \dfdx = \vec{0}\}$.
    \end{definition}
    
    The coercivity and differentiability of $f$ guarantee that any minima are within the set $S$.

\section{Component-Wise Scaled Gradient Flow}
\label{sec:gflow}
To begin this optimization idea, we first establish convergence of the component-wise scaled gradient flow in \eqref{scaled gradient flow}. It can be shown that for any objective function $f$ and any scaling function $Z$ satisfying the stated assumptions, the gradient flow IVP will converge to some $\x$ within $S$. 

\begin{theorem} \label{convergence theorem}
Under \ref{a1} - \ref{z def}, consider the component-wise scaled gradient flow IVP in \eqref{scaled gradient flow}. Then,
\begin{equation}
    \lim_{t\to\infty}\Vert \dfdxt\Vert = 0.
\end{equation}
\end{theorem}

The proof is provided in Appendix \ref{main proof}. Given Theorem \ref{convergence theorem}, the next objective is to design $Z$ that satisfies \ref{z def} and yields faster convergence than the base case of $Z=I$. To tackle this objective, we next demonstrate that \eqref{scaled gradient flow} is mathematically equal to the transient response of an electrical circuit; therefore the circuit can be used to design  $Z$.

\subsection{Equivalent Circuit of an Optimization Problem}
\label{sec:ec}
\label{sec:equivalent-circuit}

The goal of this section is to develop a circuit whose transient response is equivalent to \eqref{scaled gradient flow}. The EC is composed of $n$ sub-circuits, with each sub-circuit representing the transient waveform of a single variable, $\x_i(t)$. The diagram of a single sub-circuit is shown in Figure \ref{fig:ECorig}. Each sub-circuit is composed of two elements: a nonlinear capacitor (on the left) and a voltage controlled current source (VCCS) represented by the element on the right. Note that for multi-dimensional $\x$, the sub-circuits are coupled via the VCCSs. 

The voltage across the capacitor is defined to be $\x_i(t)$, and the capacitance of the capacitor is likewise defined as $Z_{ii}(\xt)$. Based on the voltage-current relationship of a capacitor, the current will be equal to $I_i^c(\x(t))=\Zeta_{ii}(\xt)\dxdt_i(t)$ where $c$ labels attachment to the capacitor and $i$ indexes the element of $\x$. From Kirchhoff's Current Laws (KCL) \cite{desoer2010basic}, the capacitor current must be equal to the negative of the current produced by the VCCS; therefore $Z(\xt)\dot{\x}(t)=-\dfdxt$ which is identical to \eqref{scaled gradient flow}. If the circuit reaches steady-state at some time $t'$, the capacitor no longer produces a current ($I_c(\x(t))_{t\geq t'}=0$) and therefore the node voltage $\x(t)_{t\geq t'}$ remains stationary. By KCL, the VCCS elements also produce zero current at steady-state ($\dfdxt_{t\geq t'}=\vec{0}$), implying that we have reached a point where $\dfdxt=\vec{0}$, which is defined as a critical point of the objective function.





\begin{figure}
\centering
\begin{minipage}[t]{.48\linewidth}
  \centering
    \begin{circuitikz}[scale=0.45]
    \ctikzset{label/align = rotate}
    \draw
    (0,0) to[C=\small $Z_{ii}(\xt)$] (0,3)
      to[short] (3,3)
      to[cI, l=\small $\frac{\partial f(\x(t))}{\partial \x_i}$, label distance=3pt] (3,0)
      to[short] (0,0)
    ;
    \draw (3,3) to [short,-o] (4,3) node[above]{\small $\x_i(t)$};
    \draw (1.5,0) node[ground] (){};
    \end{circuitikz}
    \caption{\small{Equivalent Circuit Model of \eqref{scaled gradient flow}} }
    \label{fig:ECorig}
\end{minipage}%
\hfill
\begin{minipage}[t]{.48\linewidth}
  \centering
   \begin{circuitikz}[scale = 0.45]
   \ctikzset{label/align = rotate}
   \draw
    (0,0) to[C, i^<= \small\color{white} $0$, ] (0,3)
      to[short] (3,3)
      to[cI, l=\small $\quad \frac{d}{dt} \frac{\partial f(\x(t))}{\partial \x_i}$, label distance=3pt] (3,0)
      to[short] (0,0)
    ;
    \draw (3,3) to [short,-o] (4,3) node[above]{\small $\dot\x_i(t)$};
    \draw (1.5,0) node[ground] (){};
    \draw (-0.5,3.8) circle [radius=0] node {\small $\frac{d}{dt}(Z(\xt)\dot\x(t))_i$};
    \end{circuitikz}
    \caption{\small{Adjoint Equivalent Circuit Model of \eqref{time deriv}}}
    \label{fig:ECadj}
\end{minipage}
\end{figure}

\subsection{Adjoint Equivalent Circuit Model}
To gain insight into the energy transfer in the equivalent circuit model, we also construct a circuit representation for the behavior of $\dxdt(t)$. Taking the time derivative of \eqref{scaled gradient flow}, we construct a circuit representation called the \emph{adjoint circuit}, as shown in Figure \ref{fig:ECadj}. Similar to the equivalent circuit model, each adjoint sub-circuit is composed of a capacitor and a VCCS element, with the node voltage now representing $\dot{\x}_i(t)$. The adjoint capacitor has a current of $ \bar{I}_i^c(\x(t)) = \frac{d}{dt}[Z(\xt) \dxdt(t)]_i$ where:

\small
\begin{align}
    \frac{d}{dt}\left(Z(\xt)\dot{\x}(t)\right) &= - \frac{d}{dt} \dfdxt = -\nabla^2 f(\xt) \frac{d \x(t)}{dt},\nonumber\\
    &= \nabla^2 f(\xt) Z(\xt)^{-1} \dfdxt.\label{time deriv}
\end{align}
\normalsize


The energy of the adjoint circuit is analyzed to provide intuition on controlling the circuit; when the adjoint circuit is at steady-state then $\dot{\x}(t)=0$, meaning that the original circuit is also at a steady-state.  The capacitor in the adjoint circuit is initially charged to  $\bar{Q}_i^c(\x(0))$ and discharges to reach steady-state. The energy stored in the capacitor is proportional to the charge of the capacitor, which in turn is:
\begin{equation}
    \bar{Q}_i^c(\xt) = \int_0^t\bar{I}_i^c(\x(t))dt = \Zeta_{ii}(\xt)\dot{\x}_i(t). \label{charge}
\end{equation}
Let $\bar{Q}_c(\xt) = [\bar{Q}^c_1(\xt), \dots, \bar{Q}^c_n(\xt)]^{\top}$.


Any component-wise scaled gradient flow problem \eqref{scaled gradient flow} satisfying \ref{a1} - \ref{z def} can thus be modeled as an EC and adjoint EC. In the next section, we will use the circuit formulation to find controls that ensure convergence steady-state.

\subsection{Control Matrix $Z$}
\label{sec:control}
\label{sec:control}
The next step is to construct a control policy for the scaling matrix $Z^{-1}$ by leveraging the circuit formulation. In the circuit sense, fast convergence to a critical point of the optimization problem is equivalent to fast convergence to a stable steady-state of the dynamical system. Thus, we must choose the nonlinear capacitances $Z$ as a function of $\xt$ to discharge the capacitors as quickly as possible. 

The squared charge stored in the adjoint circuit capacitor at some time $t$ is equal to $\Vert \bar{Q}_c(\xt)\Vert^2$, which by \eqref{charge} is equal to  $\Vert Z(\xt)\dot{\x}(t))\Vert^2$ and by \eqref{scaled gradient flow} is equal to $\Vert \nabla f(\xt)\Vert^2$. Thus to quickly \emph{dissipate} charge, we define the following optimization problem, which maximizes the negative time gradient of the charge:

\vspace{-0.5cm}
\small
\begin{align}
    &\max_{Z}- \frac{d}{dt}\Vert \bar{Q}_c(\xt) \Vert^2=\max_{Z}- \frac{d}{dt}\Vert \dfdxt \Vert^2,\\
    &=\max_Z \dfdxt^{\top}\nabla^2 f(\xt)Z^{-1}(\xt)\dfdxt.\label{max dvdt}
\end{align}
\normalsize

The optimization problem must yield a diagonal $Z$. To this end, we expand $\dfdxt$ to a diagonal matrix and shrink $Z^{-1}$ to a vector. Define $G(\xt)$ be a diagonal matrix where the diagonal elements are the gradient $G_{ii}(\xt) = \frac{\partial f(\xt)}{\partial \x_i(t)}$. Let $\mathbf{z}$ be a vector where $\mathbf{z}_i=Z_{ii}^{-1}$. Then \eqref{max dvdt} is equal to,
\begin{align}
    &=\max_{\z} \dfdxt^{\top}\nabla^2 f(\xt)G(\xt)\z-\frac{\delta}{2} \Vert \z\Vert^2,
\end{align}
where a regularization term with $\delta>0$ has been added for tractability. Taking the derivative $\frac{\partial}{\partial \z}$:
\begin{gather}
    G(\xt)\nabla^2f(\xt) \dfdxt -\delta \z \equiv \vec{0}.\\
    \z = \frac{1}{\delta} G(\xt) \nabla^2 f(\xt)\dfdxt.\label{final z}
\end{gather}
We find that truncating any $\z_i<1$ to $\z_i=1$ performs well in practice as this ensures that $Z^{-1}$ performs at least as well as the base case where $Z^{-1}=I$, i.e. uncontrolled gradient flow. This truncation choice also ensures positivity and invertibility of $Z$. For some $\delta>0$, the final construction of the control matrix $Z^{-1}(\xt)$ is thus:
\begin{equation}
    Z_{ii}^{-1}(\xt) = \max\{\delta^{-1} [G(\xt) \nabla^2 f(\xt)\dfdxt]_i,1\}.\label{z true}
\end{equation}

Equation \eqref{z true} is a second order method as it uses Hessian information. Given the gradient and Hessian, the per-iteration computation complexity is $\mathcal{O}(n^2)$ to evaluate $Z_{ii}^{-1}$ at a specific $\xt$\footnote{See Appendix \ref{complexity full}}. In comparison to Newton methods, this $Z$ does not require a Hessian inversion step.

We show that \eqref{z true} satisfies Assumption \ref{z def} in the Appendix \ref{zbounded}. Note that while $\delta$ is a hyperparamter, in practice we find it is has little impact on performance with normalization of $Z^{-1}$, and can often be set to $\delta=1$\footnote{See Appendix \ref{more robust}}.

\subsection{Approximate Control Scheme $\widehat{Z}$}
\label{sec:acontrol}

The proposed control scheme may be used directly as in \eqref{z true}; however, it requires computation of the full Hessian. For applications in which the Hessian in unavailable or expensive, such as machine learning, we present an approximation that may be computed from only gradient information.

Define the finite difference approximation, i.e. the first order Taylor expansion, of the optimization trajectory $\frac{d}{dt}\dfdxt$,
\begin{align}
    \ahat(\xt)\triangleq \frac{\dfdxt - \nabla f(\x (t-\Delta t))}{\Delta t}. \label{ahat}
\end{align}
The limit as $\Delta t\to 0$ is exactly equal to desired quantity $\frac{d}{dt}\nabla f(\xt)$, making this an apt approximation for small $\Delta t$. Based on \eqref{ahat}, approximation for $Z^{-1}$ can be defined:
\begin{equation}
     \widehat{Z}_{ii}(\xt)^{-1}=\sqrt{\max\{[-\delta^{-1}G(\xt)\ahat(\xt)]_i,1\}}. \label{z approx}
\end{equation}

The derivation is shown in Appendix \ref{z approx deriv}. We verify that \eqref{z approx} satisfies Assumption \ref{z def} in Appendix \ref{zhat bounded}. Given the gradients, the per-iteration computation complexity is $\mathcal{O}(n)$ to evaluate $\widehat{Z}_{ii}^{-1}$ at a specific $\xt$\footnote{See Appendix \ref{complexity approx}}. 

We form $\widehat{Z}$ by approximating the trajectory of the optimization variable at a specific $t$ and solving for $Z$ as a function of $\ahat$; we do not estimate the Hessian matrix, and therefore $\widehat{Z}$ is not a quasi-Newton method. The proposed iterative update using \eqref{z approx} is a scaled gradient or step-size normalized descent method, similar in spirit to normalized gradient flow (in continuous time) and step-size or momentum scaling methods methods such as Adam, RMSprop, and Adagrad.

\section{Circuit Simulation Inspired Discretization}
\label{sec:discret}
Armed with a continuous ODE that converges to a critical point, the next step is to find the value of the steady-state. Three obvious methods from the circuits world are: physical circuitry, commercial circuit simulators, and circuit-based discretization techniques. Building a physical circuit to solve an EC is an attractive long-term goal as it could automatically solve large classes of problems, i.e. neural networks, and could eliminate computational issues of standard iterative solvers. Finding the steady-state value would involve building an analog VCCS for the objective function, turning on power, and then measuring the voltage at the node representing $\x(t)$. Despite these advantages, building a physical analog circuit is difficult because fabricating an application-specific device leads to issues in analog computing such as energy inefficiency, noise and process variations.

The second option is to use circuit simulating software, draw a schematic of the EC, and observe the transient waveform to steady-state. Decades of research have yielded top-of-the-line industry tools such as SPICE \cite{nagel1975spice2}, LTSpice \cite{ltspice} and MultiSim \cite{multisim}. For example, we can build an EC to minimize $f(\x) = \frac{5}{2}\x^2 + \x$ in MultiSim. The schematic and transient simulation for the EC model is shown in Figure \ref{fig:multisim}. Simulating the transient waveform with an initial condition of $\x(0)=1V$, i.e. initial guess of $\x_0 = 1$, we see that the node-voltage reaches a steady-state of $\x=-0.2V$ at 100$\mu s$. It can be verified that $\x^*=-0.2$ is the optimum.

\ifarxiv
\begin{figure}
    \centering
    \includegraphics[width=0.6\linewidth]{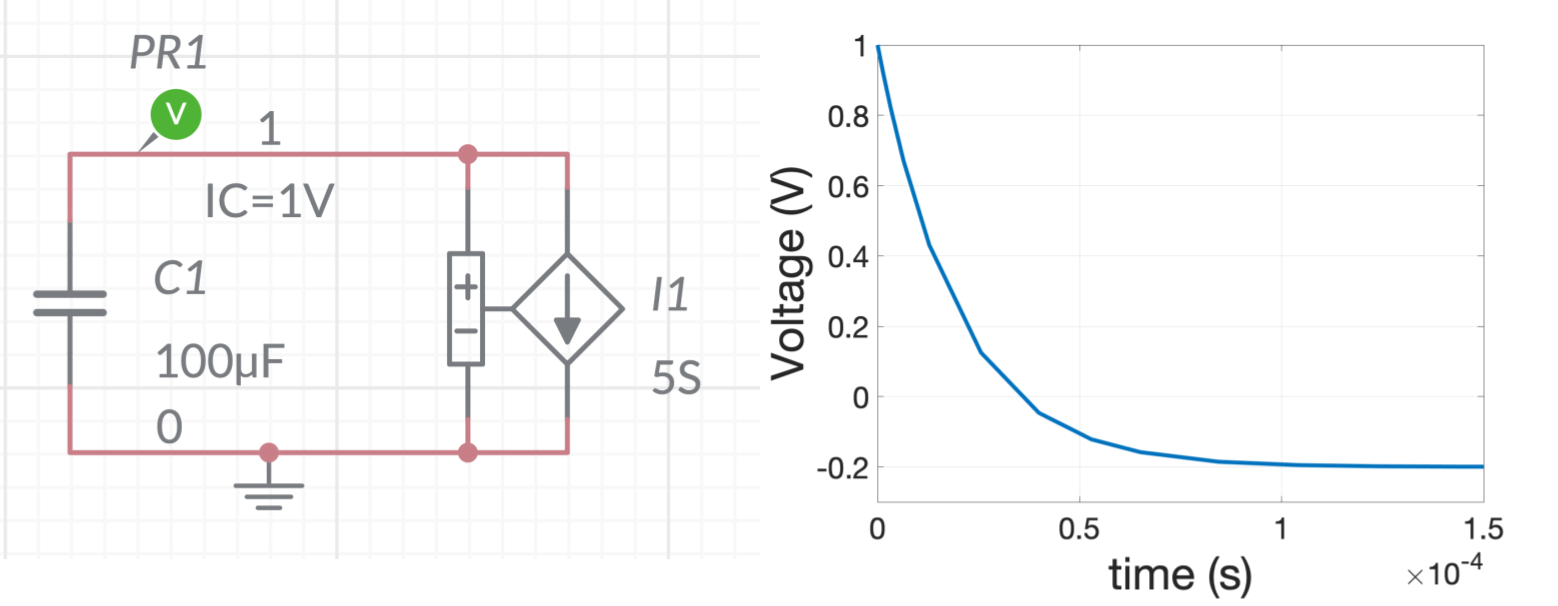}
    \caption{Schematic and EC model to minimize $f(\x) = \frac{5}{2}\x^2 + \x$ via gradient flow with MultiSim.}
    \label{fig:multisim}
\end{figure}
\fi

In this work, we borrow the technical backend to SPICE's transient simulation method and adapt it for optimization.

\subsection{Error-Aware Discretization} \label{our disc}

SPICE solves for the transient waveform by approximating the state of the system in discrete time-steps. The accuracy of the approximation relies on the type of numerical integration and size of the time-step, $\Delta t$.  SPICE maximizes $\Delta t$ subject to accuracy and stability conditions; parallels may be drawn to backtracking line search in optimization. 

We propose a FE discretization, but we adapt the step size to the shape of the waveform based on circuit simulation ideas: (1) the initial step size is selected according to a passivity metric, and (2) this guess is refined to satisfy an error metric. These techniques are widely used in circuit simulation, but to the best of our knowledge have not been applied to optimization theory. These techniques will argue why the simple FE excels over other discretization schemes. 

For any $\Delta t$ a discretization technique aims to approximate the solution integral $\x(t+\Delta t) -\x(t)=\int_{t}^{t+\Delta t} \dot{\x}(t) dt.$ Circuit simulation tools such as SPICE choose step sizes $\Delta t$ by controlling the accuracy of their approximation via the \emph{local truncation error} (LTE), which measures the goodness of the approximation \cite{pillage1995electronic}. SPICE yields accurate simulation as it chooses a tiny LTE, normally on the order of $10^{-5}$ paired with an implicit integration technique. Compared to explicit methods (FE), implicit methods are much more accurate but are far slower. Implicit methods require the costly solving of nonlinear equations at each step, plus a smaller LTE naturally means more steps must be taken. For example, we can solve the EC model for a Rosenbrock test function \cite{testfunc} using an implicit discretization method (Matlab ODE45). With the same LTE tolerance, the implicit solver has a wall-clock time that is 3.62x \emph{slower} than an explicit integration technique to reach the same local optimum.

Our goal is to reach a steady-state \emph{fast}. We do not need to exactly find the solution trajectory of the ODE; we only want its fixed point. As a result, we relax the accuracy constraints imposed by SPICE, allowing a larger LTE tolerance and the selection of easy-to-compute methods like FE. Furthermore, the LTE for FE can be approximated from only gradient information as described in chapter 4 of \cite{pillage1995electronic}  :
\begin{align}
   LTE &= 0.5 \Delta t |I_i^c(\x(t+\Delta t)) - I_i^c(\x(t))|,\nonumber\\
   &= 0.5 \Delta t |\dfdxt - \nabla f(\x(t+\Delta t))|.
    \label{eq:lte}
\end{align}
\normalsize

Explicit integration techniques like FE, however, often encounter issues of compounding LTE, causing numerical instability. It is generally difficult for numerical solvers to enforce a stability check for explicit integration methods as the final steady-state is unknown \cite{pillage1995electronic}. Numerical instability issues may be observed when solving the EC in a commercial numerical solver with FE.

In this formulation, however, we know that the trajectory should satisfy $f(\x(t))>f(\x(t+\Delta t))$ for $\x(t)\neq \x(t+\Delta t)$. As a stability check, we can simply exploit the optimization formulation and require that it monotonically decreases in time and goes towards a local minimum. For example, the EC model for the Rosenbrock function is solved using an explicit integration method in MATLAB (MATLAB ODE23), with an initial condition of $\x(0)=[0,0]$. The trajectory of the MATLAB ODE23 discretization, shown in Figure \ref{fig:rosenbrock_ode_comparison}, is unable to converge to a steady-state for the same LTE tolerance as ECCO.

\ifarxiv
\begin{figure}
    \centering
    \includegraphics[width=0.4\linewidth]{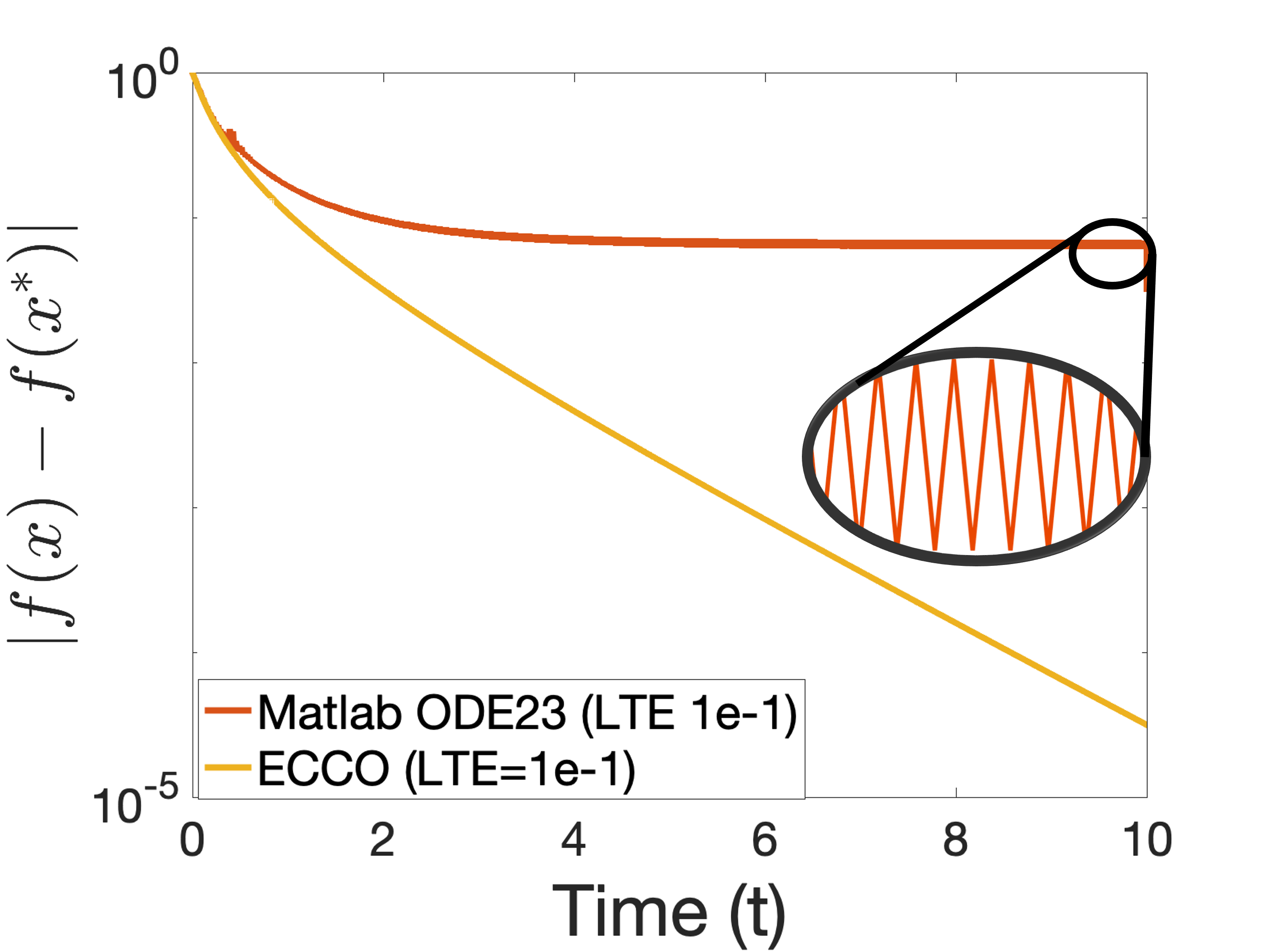}
    \caption{\small{Trajectory of Rosenbrock Function with initial condition $\x(0)=[0,0]$ using Matlab ODE23 and FE-EATSS. Note Matlab ODE23 oscillates while FE+EATSS reaches steady-state.}}
    \label{fig:rosenbrock_ode_comparison}
\end{figure}
\fi

Our LTE-aware discretization is presented in Algorithm \ref{time step algo}.  A $\Delta t$ is selected at every iteration subject to two criteria: (1) satisfying an acceptable LTE tolerance level $\eta>0$, and (2) monotonically decreasing the objective function value\footnote{Note that condition (1) is taken from circuit simulation theory and condition (2) is commonly used in traditional optimization theory as a version of the Armijo condition. This work attempts to unite the best ideas from both fields.}. The time step is selected by choosing an initial guess $(\Delta t)_0$, and then scaling up or scaling down the guess until the largest time step that satisfies both conditions is found. 

The first initial time step guess is borrowed from circuit simulation literature.  A $(\Delta t)_0$ guess that may satisfy the LTE criteria for FE discretization is known to be:
\begin{equation}
    (\Delta t)_0 = \frac{\x(0)^{\top}\nabla f(\x(0))}{ \sum_i Z_{ii}(\x(0))^{-1}\frac{\partial f}{\partial \x_i}\big|_{\x(0)}}. \label{tfe-naive}
\end{equation}

This time step is based on linearizing the circuit about the operating point at time $t$ and choosing a time step that ensures that the linear circuit, composed of the capacitor and the linearized VCCS elements, is passive. Further explanation on this choice of initial $\Delta t$ is provided in \cite{rohrer1981passivity}. For subsequent time steps, the last executed choice of $\Delta t$ is used as the next initial guess \footnote{We find in simulation that \eqref{tfe-naive} often yielded step sizes that satisfied the LTE and monotonicity checks within a few inner loop iterations. See Appendix \ref{backtrack} for experimental details.}.


\section{Main Algorithm}
\label{sec:Algo}
In this section, we summarize how to use scaled gradient flow and error-aware FE discretization to solve optimization problems algorithmically: Equivalent Circuit Controlled Optimization (ECCO) in Algorithm \ref{bigger algo}. The output of the algorithm is a critical point of \eqref{prob}. A user must decide to implement either the second order \eqref{z true} or the first order control scheme \eqref{z approx}. If the Hessian can be easily computed then the second order method is preferable because no error will be incurred from the trajectory approximation. For applications in which the Hessian is slow to query, such as in machine learning, then the speed of the first order option will likely out-weigh the effects of the approximation error.

We emphasize that the proposed algorithms are \emph{robust} to the hyperparameters. The convergence of the continuous-time trajectory is independent of $\delta$, so while its value can slow the convergence speed (large $\delta$ recovers general gradient flow, small $\delta$ will be "ignored" with normalization) it cannot cause divergence. $(\Delta t)_0$, $\alpha$, and $\beta$ only affect the speed of the time search step; the LTE tolerance is the only discretization hyperparameter that will affect the number of steps taken as it bounds the maximum deviation of the approximation from the true trajectory. As the LTE condition is coupled with the stability check to ensure that $f(\x(t))>f(\x(t+\Delta t))$ for $\x(t)\neq \x(t+\Delta t)$, $\eta$ cannot cause divergence. In simulation we find $\eta$ can be changed by orders of magnitude with little effect on the optimization trajectory. We find that $\eta=0.1$, $\delta=1$, and normalizing $Z^{-1}$ at each iteration works well\footnote{In circuits literature, $\eta$ is commonly chosen as 0.001 \cite{ltspice}, but we recommend larger $\eta$ based on the discussion in \S\ref{our disc}.} .

\begin{algorithm}[h]\small
 \SetAlgoLined
 \LinesNumbered
 \SetKwInOut{Input}{Input}
 \Input{$f(\cdot)$, $\nabla_{\x} f(\cdot)$, $\x(0)$, $\delta>0$, $\alpha \in (0,1)$, $\beta>1$, $\eta> 0$, $\epsilon> 0$}
 \KwResult{$\x\in S$}
 \SetKwProg{Function}{function}{}{end}
 $t\gets 0$\;
 
 $(\Delta t)_0$ according to \eqref{tfe-naive}\;
 
 \SetKwRepeat{Do}{do}{while}
     \Do{$\|f(\x(t-\Delta t))-f(\x(t))\|> \epsilon$}
     { 
     Choose $Z_{ii}^{-1}(\x(t))$ according to \eqref{z true} or \eqref{z approx}\;
     
    $\Delta t \gets$ \Call{EATSS}{$f(\cdot)$,  $\nabla f(\cdot)$, $\x(t)$, $Z(\xt)$, $(\Delta t)_t$, $\alpha\in(0,1)$, $\beta>1$, $\eta>0$}\;
    
    $\x(t+\Delta t)=\x(t)-\Delta t Z(\xt)^{-1}\dfdxt$\;
    
    Take step $t\gets t+\Delta t$\;
    
    $(\Delta t)_t\gets \Delta t$\;
     }
 \caption{Equivalent Circuit Controlled Optimization (ECCO)}\label{bigger algo}
\end{algorithm}

\begin{algorithm}[h]\small
\setcounter{AlgoLine}{0}
 \SetAlgoLined
 \LinesNumbered
 \SetKwInOut{Input}{Input}
 \Input{$f(\cdot)$,  $\nabla f(\cdot)$, $\x(t)$, $Z(\xt)$, $(\Delta t)_t$, $\alpha\in(0,1)$, $\beta>1$, $\eta>0$}
 \KwResult{$\Delta t$}
 \SetKwProg{Function}{function}{}{end}
 $\Delta t \gets (\Delta t)_t$\;
 
 $\x(t+\Delta t)=\x(t)-\Delta t Z(\xt)^{-1}\nabla f(\x(t))$\;
 
 $LTE = 0.5\Delta t|\dfdxt - \nabla f(\x(t+\Delta t))|$\;
 
 \While{$\max(LTE) < \eta$ and $f(\x(t+\Delta t)) < f(\x(t))$}{
 $\Delta t = \beta \Delta t$\;
 }
 \While{$\max(LTE) > \eta$ or $f(\x(t+\Delta t)) > f(\x(t))$}{
 $\Delta t = \alpha \Delta t$\;
 }
 \caption{Error Aware Time Step Search (EATSS)}\label{time step algo}
\end{algorithm}

\section{Simulations}
\label{sec:Sims}
We demonstrate the ECCO methodology in a variety of examples. We consider the following cases: (1) full Hessian ECCO $Z$ \eqref{z true}, (2) first order ECCO $\widehat{Z}$ \eqref{z approx}, and comparison methods (3) baseline gradient flow, i.e. $Z\equiv I$ for all $\x$, paired with FE+EATSS discretization, (4) Adam optimization, and (5) gradient descent with backtracking line search and the Armijo condition \cite{wolfe1969convergence}. Method (1) versus method (2) demonstrates the performance loss without the Hessian, and (1) versus (3) demonstrates the improvement in performance when designing $Z$ for fast convergence. Note that baseline gradient flow with standard FE discretization is equivalent to gradient descent, which may be considered as method (4); this means that (3) and (4) show the improvement provided by FE+EATSS without any trajectory control via $Z$. In all experiments, ECCO used $\delta=1$, $\alpha = 0.9$, $\beta=1.1$, $\eta=0.1$. The hyperparameters for all comparison methods were tuned on every experiment. 
 
\subsection{Test Functions}
\ifarxiv
\begin{figure}
    \centering
    \includegraphics[width=0.7\linewidth]{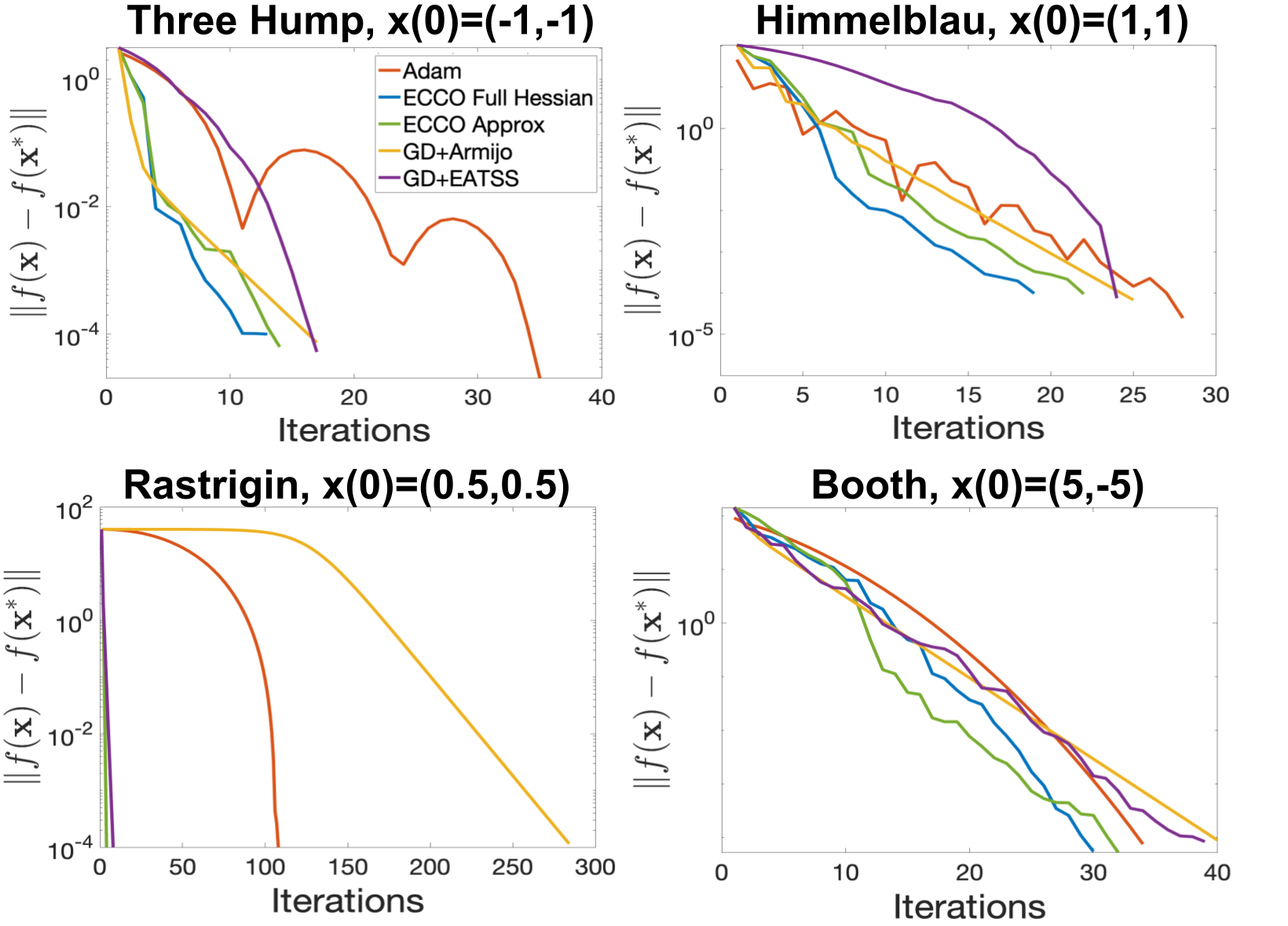}
    \caption{\small{Comparison of Optimization Methods on Test Functions. Note that in the Rastrigin test function, the plots for the ECCO full Hessian case and the ECCO approximate Hessian case overlap. ECCO consistantly converges faster than baseline gradient flow, gradient descent with line search, and Adam.}}
    \label{fig:results}
\end{figure}
\fi

To evaluate ECCO's convergence properties, we first optimize convex and nonconvex test functions in order to compare performance to known local minima. See Appendix \ref{all toy graphs} for the full experiment details and results. Convergence was defined as $\x_k$ such that $\Vert f(\x_k)-f(\x_{k+1})\Vert <1e-4$. The results are shown in Figure \ref{fig:results}. First note that $\widehat{Z}$ appears to well approximate  $Z$. Second, these results demonstrate that both $Z$ and $\widehat{Z}$ provide faster convergence than baseline gradient flow where all the scaling factors are one. Third, note that the FE+EATSS discretization generally provides a speed up over FE with backtracking line search and the Armijo condition. Finally, note that ECCO provided strong results without needing hyperparameter tuning.


\subsection{Power Systems Optimization} 
We next apply ECCO to a power systems optimization problem to demonstrate its performance on an applied example. In this experiment, we optimize for the state variables in a stressed power grid network (a loading factor of 1.5) that has undergone a line contingency (line outage from bus 8 to 30). The resulting post-contingency network is infeasible and we optimize for the state variables, (which include bus voltage magnitude, angles and generator reactive powers) that account for the smallest infeasibility in the network. The solution represents the state of the network with the fewest violations of the power flow constraints. This information is vital to determine the location and amount of additional assets required to satisfy a contingency scenario for infrastructure planning applications. 

The optimization problem is $min_\x \frac{1}{2} \| f(\x) \| ^2$ where $f(\x)$ defines the power-flow constraints at each bus (defined in \cite{FOSTER2022108486}). The objective is a multi-modal, non-convex function and the dimensionality is 291.
We optimize this objective using ECCO with a full Hessian \eqref{z true}, the approximate ECCO implementation \eqref{z approx}, GD with Armijo Line Search, BFGS \cite{fletcher2013practical}, SR1 \cite{conn1991convergence} and Newton-Raphson \cite{fletcher2013practical}. The results are displayed in Figure \ref{fig:power}. In this example, Newton-Raphson diverged immediately and is not shown on the graph. Full Hessian ECCO converges to a local minimum, as do approximate ECCO, BFGS, SR1, and GD; however, full ECCO found the best local optimum with the lowest total infeasibility. Approximate ECCO achieved the same local minimum as BFGS and SR1 in less than a third of the iterations. In terms of wall-clock time, ECCO reached the local minimum in 3.74s, where as BFGS took 7.89s and SR1 took 9.14s. These results demonstrate that ECCO can perform well on general nonconvex multi-modal optimization problems.

\ifarxiv
\begin{figure}
    \centering
    \includegraphics[width=0.4\linewidth]{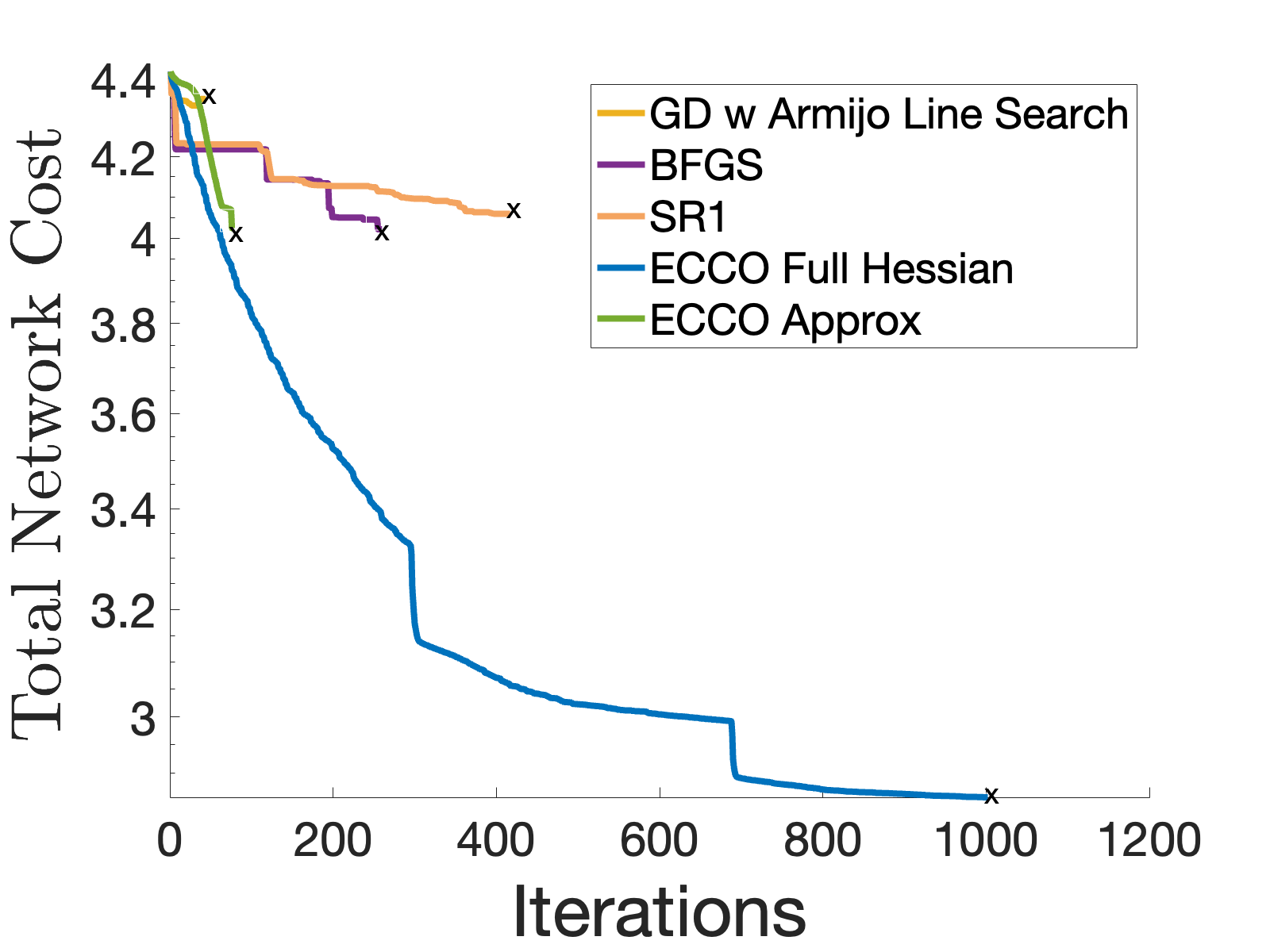}
    \caption{\small{Comparison of optimization methods on optimizing a post-contingency power grid. }}
    \label{fig:power}
\end{figure}
\fi

\subsection{Training a Deep Neural Network for Classification} \label{sec:nn}
We next use ECCO on a machine learning application by training a neural network to classify MNIST data \cite{deng2012mnist}. See the Appendix \ref{more nn} for details on this experiment. We compared first order ECCO \eqref{z approx} to stochastic gradient descent (SGD) with a fixed step size, Adam, and RMSProp\footnote{A grid search was used to find hyperparameters for the comparison methods. See Appendix \ref{hyper search} for details} using a minibatch size of 1000. ECCO was adapted for a batch implementation in that every gradient was replaced with the sample gradient over the observed batch.


Figure \ref{fig:supp_training_loss_mnist} shows the results of training this experiment over 200 epochs; note that all four methods reach a similar training loss and have a similar test classification accuracy ($96\%$). See Appendix \ref{nn wall clock} for a discussion of wall clock times.

To test robustness, we perturbed the optimally tuned hyperparameter vectors within a normalized ball of radius $\varepsilon=0.1$. The hyperparameter values were sampled from a uniform distribution as $\tilde\theta \sim U(\max\{{\theta}^*-\frac{\varepsilon}{{\theta}^*}, \underline{\theta}\},\min\{{\theta}^*+\frac{\varepsilon}{{\theta}^*},\bar{\theta} \}) $, where ${\theta}^*$ was the optimal value as found by grid search. Note that this perturbation was bounded such that $\tilde\theta$ was always within the domain $[\underline{\theta} ,\bar{\theta}]$ of the hyperparameter. The network was then trained with $\tilde{\theta}$ and the test accuracy after 200 epochs was recorded in Figure \ref{fig:robust in paper}. We find that while Adam, SGD, and RMSProp are disrupted by the hyperparameter perturbation, ECCO reliably trained in every test. In fact, we found $\eta$ could be changed by three orders of magnitude and still yield classification accuracies over 95\%. See Appendix \ref{more robust} for more experiments. 

These experiments demonstrate that ECCO can achieve convergence behavior on par with Adam without comparable computation and data requirements to select hyperparameters, i.e. cross-validation. This means ECCO is better than Adam at problem generalization and  distribution shift.




\ifarxiv
\begin{figure}
    \centering
    \includegraphics[width=0.6\linewidth]{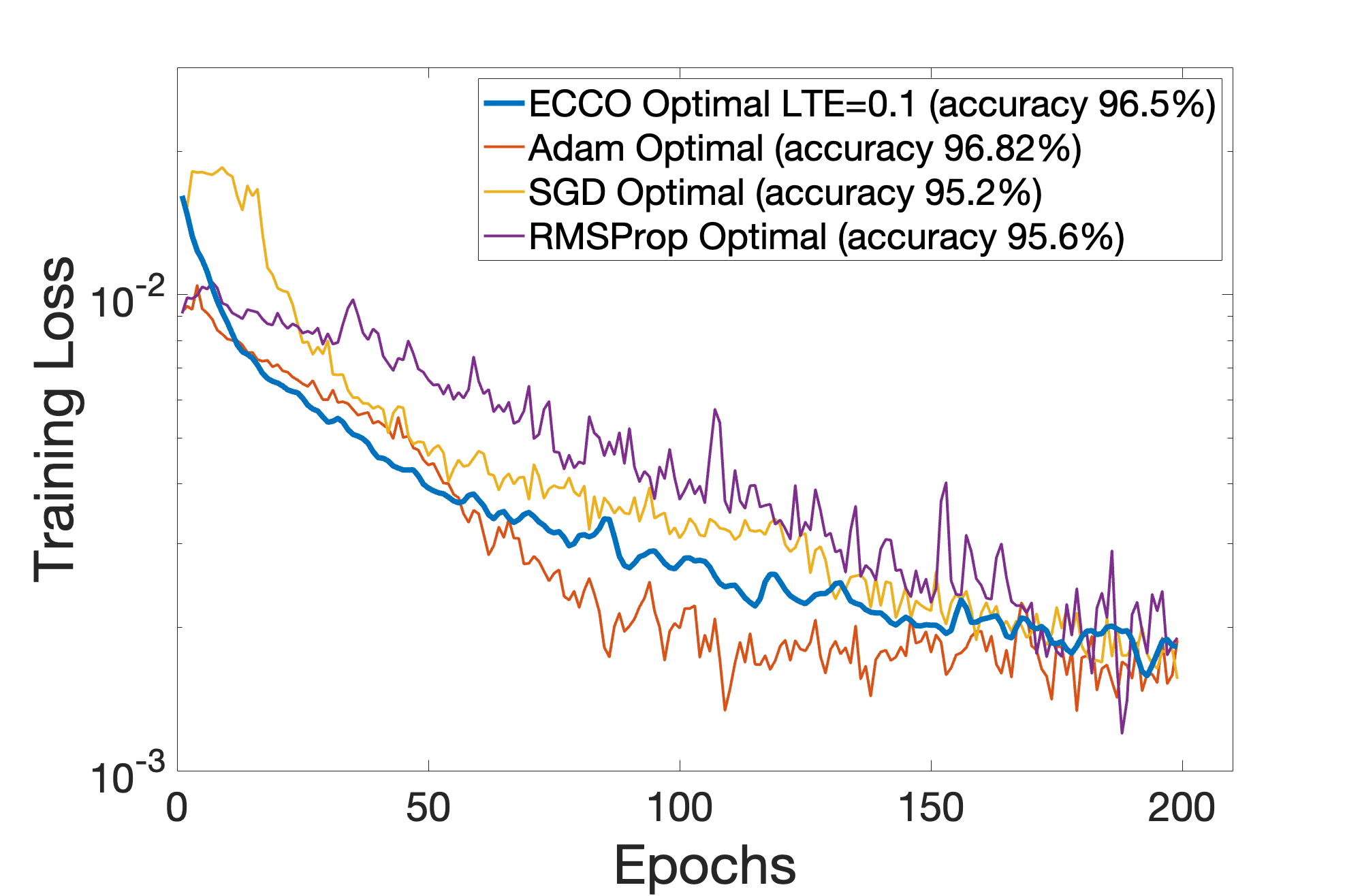}
    \caption{\small{Training of neural network for MNIST data using optimally tuned first order ECCO \eqref{z approx}, Adam, GD and RMSProp.}}
    \label{fig:supp_training_loss_mnist}
\end{figure}
\fi

\ifarxiv
\begin{figure}
    \centering
    \includegraphics[width=0.5\linewidth]{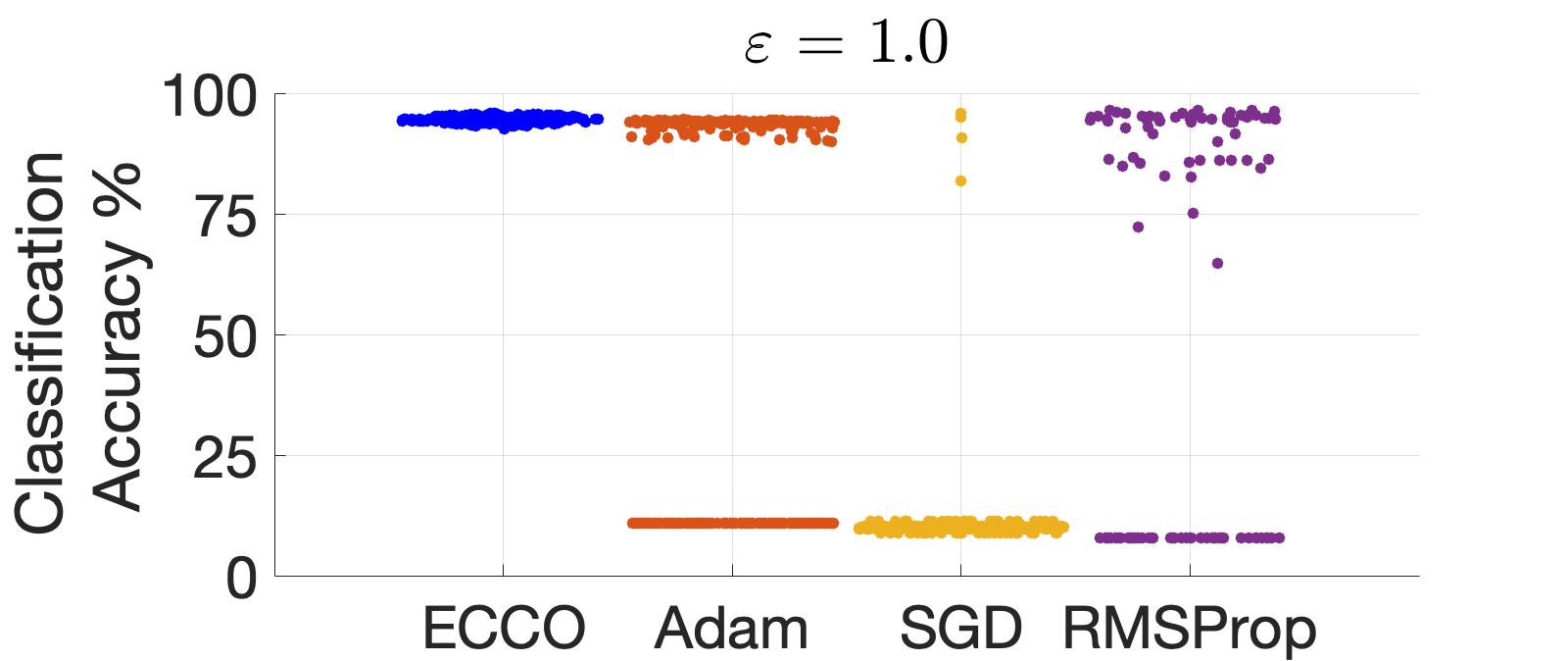}
    \caption{\small{Classification accuracy comparison using perturbed hyperparameters. Note that ECCO always successfully trained the neural network; the comparison methods were not reliable.}}
    \label{fig:robust in paper}
\end{figure}
\fi

\section{Conclusion}
\label{sec:Conc}
In this paper we propose ECCO, a novel adaptive optimization algorithm that solves unconstrained scaled gradient flow problems and achieves fast convergence by controlling the optimization trajectory shape and the discretization step sizes. The proposed method leverages the equivalence with the transient response of an electrical circuit, allowing for the use of circuit theory to solve the optimization problem. Simulation results show not only performance comparable or exceeding state-of-the-art algorithms, such as Adam, but also superior robustness to initialization and hyperparameters variations. This is particularly relevant for neural networks, where, as shown in our study, comparable algorithms can show large fluctuations in accuracy as a result of small changes in the choice of hyperparameters. We believe that this overlooked connection between this important class of optimization problems and circuit theory is key for innovation opportunities, and we plan to extend this theory further, starting with applications to batch methods and discontinuous objective functions.

\pagebreak

\nocite{langley00}

\bibliography{example_paper}

\begin{thebibliography}{41}
\providecommand{\natexlab}[1]{#1}
\providecommand{\url}[1]{\texttt{#1}}
\expandafter\ifx\csname urlstyle\endcsname\relax
  \providecommand{\doi}[1]{doi: #1}\else
  \providecommand{\doi}{doi: \begingroup \urlstyle{rm}\Url}\fi

\bibitem[Aflak(2020)]{nnet_code}
Aflak, O.
\newblock Python neural networks, 2020.
\newblock URL \url{https://github.com/OmarAflak/python-neural-networks}.

\bibitem[Amid \& Warmuth(2020)Amid and Warmuth]{amid2020reparameterizing}
Amid, E. and Warmuth, M.~K.
\newblock Reparameterizing mirror descent as gradient descent.
\newblock \emph{Advances in Neural Information Processing Systems},
  33:\penalty0 8430--8439, 2020.

\bibitem[Andrei(2004)]{andrei2004gradient}
Andrei, N.
\newblock Gradient flow algorithm for unconstrained optimization.
\newblock \emph{ICT Technical Report}, 2004.

\bibitem[Andrei(2008)]{andrei2008unconstrained}
Andrei, N.
\newblock An unconstrained optimization test functions collection.
\newblock \emph{Adv. Model. Optim}, 10\penalty0 (1):\penalty0 147--161, 2008.

\bibitem[Attouch \& Cominetti(1996)Attouch and Cominetti]{attouch1996dynamical}
Attouch, H. and Cominetti, R.
\newblock A dynamical approach to convex minimization coupling approximation
  with the steepest descent method.
\newblock \emph{Journal of Differential Equations}, 128\penalty0 (2):\penalty0
  519--540, 1996.

\bibitem[Barrett \& Dherin(2020)Barrett and Dherin]{barrett2020implicit}
Barrett, D.~G. and Dherin, B.
\newblock Implicit gradient regularization.
\newblock \emph{arXiv preprint arXiv:2009.11162}, 2020.

\bibitem[Behrman(1998)]{behrman1998efficient}
Behrman, W.
\newblock \emph{An efficient gradient flow method for unconstrained
  optimization}.
\newblock stanford university, 1998.

\bibitem[Boyd(2021)]{boyd2021distributed}
Boyd, S.
\newblock Distributed optimization: Analysis and synthesis via circuits, 2021.

\bibitem[Brown \& Bartholomew-Biggs(1989)Brown and
  Bartholomew-Biggs]{brown1989some}
Brown, A.~A. and Bartholomew-Biggs, M.~C.
\newblock Some effective methods for unconstrained optimization based on the
  solution of systems of ordinary differential equations.
\newblock \emph{Journal of Optimization Theory and Applications}, 62\penalty0
  (2):\penalty0 211--224, 1989.

\bibitem[Conn et~al.(1991)Conn, Gould, and Toint]{conn1991convergence}
Conn, A.~R., Gould, N.~I., and Toint, P.~L.
\newblock Convergence of quasi-newton matrices generated by the symmetric rank
  one update.
\newblock \emph{Mathematical programming}, 50\penalty0 (1-3):\penalty0
  177--195, 1991.

\bibitem[Cort{\'e}s(2006)]{cortes2006finite}
Cort{\'e}s, J.
\newblock Finite-time convergent gradient flows with applications to network
  consensus.
\newblock \emph{Automatica}, 42\penalty0 (11):\penalty0 1993--2000, 2006.

\bibitem[Deng(2012)]{deng2012mnist}
Deng, L.
\newblock The mnist database of handwritten digit images for machine learning
  research.
\newblock \emph{IEEE Signal Processing Magazine}, 29\penalty0 (6):\penalty0
  141--142, 2012.

\bibitem[Desoer(2010)]{desoer2010basic}
Desoer, C.~A.
\newblock \emph{Basic circuit theory}.
\newblock Tata McGraw-Hill Education, 2010.

\bibitem[Fletcher(2013)]{fletcher2013practical}
Fletcher, R.
\newblock \emph{Practical methods of optimization}.
\newblock John Wiley \& Sons, 2013.

\bibitem[Foster et~al.(2022)Foster, Pandey, and Pileggi]{FOSTER2022108486}
Foster, E., Pandey, A., and Pileggi, L.
\newblock Three-phase infeasibility analysis for distribution grid studies.
\newblock \emph{Electric Power Systems Research}, 212:\penalty0 108486, 2022.
\newblock ISSN 0378-7796.
\newblock \doi{https://doi.org/10.1016/j.epsr.2022.108486}.
\newblock URL
  \url{https://www.sciencedirect.com/science/article/pii/S0378779622006125}.

\bibitem[Franca et~al.(2018)Franca, Robinson, and Vidal]{franca2018admm}
Franca, G., Robinson, D., and Vidal, R.
\newblock Admm and accelerated admm as continuous dynamical systems.
\newblock In \emph{International Conference on Machine Learning}, pp.\
  1559--1567. PMLR, 2018.

\bibitem[Helmke \& Moore(2012)Helmke and Moore]{helmke2012optimization}
Helmke, U. and Moore, J.~B.
\newblock \emph{Optimization and dynamical systems}.
\newblock Springer Science \& Business Media, 2012.

\bibitem[Himmelblau(1972)]{0070289212}
Himmelblau, D.~M.
\newblock \emph{Applied Nonlinear Programming}.
\newblock McGraw-Hill, 1972.
\newblock ISBN 0070289212.

\bibitem[Hu \& Lessard(2017)Hu and Lessard]{hu2017dissipativity}
Hu, B. and Lessard, L.
\newblock Dissipativity theory for nesterov’s accelerated method.
\newblock In \emph{International Conference on Machine Learning}, pp.\
  1549--1557. PMLR, 2017.

\bibitem[Hustig-Schultz \& Sanfelice(2019)Hustig-Schultz and
  Sanfelice]{hustig2019robust}
Hustig-Schultz, D.~M. and Sanfelice, R.~G.
\newblock A robust hybrid heavy ball algorithm for optimization with high
  performance.
\newblock In \emph{2019 American Control Conference (ACC)}, pp.\  151--156.
  IEEE, 2019.

\bibitem[Kingma \& Ba(2014)Kingma and Ba]{kingma2014adam}
Kingma, D.~P. and Ba, J.
\newblock Adam: A method for stochastic optimization.
\newblock \emph{arXiv preprint arXiv:1412.6980}, 2014.

\bibitem[Latz(2021)]{latz2021analysis}
Latz, J.
\newblock Analysis of stochastic gradient descent in continuous time.
\newblock \emph{Statistics and Computing}, 31\penalty0 (4):\penalty0 1--25,
  2021.

\bibitem[Lin et~al.(2016)Lin, Ren, and Farrell]{lin2016distributed}
Lin, P., Ren, W., and Farrell, J.~A.
\newblock Distributed continuous-time optimization: nonuniform gradient gains,
  finite-time convergence, and convex constraint set.
\newblock \emph{IEEE Transactions on Automatic Control}, 62\penalty0
  (5):\penalty0 2239--2253, 2016.

\bibitem[{Mike Engelhard}()]{ltspice}
{Mike Engelhard}.
\newblock Ltspice.
\newblock URL
  \url{https://www.analog.com/en/design-center/design-tools-and-calculators/ltspice-simulator.html}.

\bibitem[Muehlebach \& Jordan(2019)Muehlebach and
  Jordan]{pmlr-v97-muehlebach19a}
Muehlebach, M. and Jordan, M.
\newblock A dynamical systems perspective on {N}esterov acceleration.
\newblock In Chaudhuri, K. and Salakhutdinov, R. (eds.), \emph{Proceedings of
  the 36th International Conference on Machine Learning}, volume~97 of
  \emph{Proceedings of Machine Learning Research}, pp.\  4656--4662. PMLR,
  09--15 Jun 2019.
\newblock URL \url{https://proceedings.mlr.press/v97/muehlebach19a.html}.

\bibitem[Muehlebach \& Jordan(2021)Muehlebach and
  Jordan]{muehlebach2021optimization}
Muehlebach, M. and Jordan, M.~I.
\newblock Optimization with momentum: Dynamical, control-theoretic, and
  symplectic perspectives.
\newblock \emph{Journal of Machine Learning Research}, 22\penalty0
  (73):\penalty0 1--50, 2021.

\bibitem[Murray et~al.(2019)Murray, Swenson, and Kar]{murray2019revisiting}
Murray, R., Swenson, B., and Kar, S.
\newblock Revisiting normalized gradient descent: Fast evasion of saddle
  points.
\newblock \emph{IEEE Transactions on Automatic Control}, 64\penalty0
  (11):\penalty0 4818--4824, 2019.

\bibitem[Nagel \& Rohrer(1971)Nagel and Rohrer]{nagel1971computer}
Nagel, L. and Rohrer, R.
\newblock Computer analysis of nonlinear circuits, excluding radiation
  (cancer).
\newblock \emph{IEEE Journal of Solid-State Circuits}, 6\penalty0 (4):\penalty0
  166--182, 1971.

\bibitem[Nagel(1975)]{nagel1975spice2}
Nagel, L.~W.
\newblock Spice2: A computer program to simulate semiconductor circuits.
\newblock \emph{Ph. D. dissertation, University of California at Berkeley},
  1975.

\bibitem[{National Instrument Corporation}()]{multisim}
{National Instrument Corporation}.
\newblock Multisimlive.
\newblock URL \url{https://www.multisim.com/}.

\bibitem[Peressini et~al.(1988)Peressini, Sullivan, and
  Uhl~Jr]{peressini1988mathematics}
Peressini, A.~L., Sullivan, F.~E., and Uhl~Jr, J.~J.
\newblock \emph{The mathematics of nonlinear programming}.
\newblock Springer-Verlag, 1988.

\bibitem[Pillage et~al.(1995)Pillage, Rohrer, and
  Visweswariah]{pillage1995electronic}
Pillage, L., Rohrer, R., and Visweswariah, C.
\newblock \emph{Electronic Circuit and System Simulation Methods}.
\newblock McGraw-Hill, 1995.
\newblock ISBN 9780070501690.
\newblock URL \url{https://books.google.ca/books?id=uZZTAAAAMAAJ}.

\bibitem[Polyak \& Shcherbakov(2017)Polyak and Shcherbakov]{polyak2017lyapunov}
Polyak, B. and Shcherbakov, P.
\newblock Lyapunov functions: An optimization theory perspective.
\newblock \emph{IFAC-PapersOnLine}, 50\penalty0 (1):\penalty0 7456--7461, 2017.

\bibitem[Rohrer \& Nosrati(1981)Rohrer and Nosrati]{rohrer1981passivity}
Rohrer, R. and Nosrati, H.
\newblock Passivity considerations in stability studies of numerical
  integration algorithms.
\newblock \emph{IEEE transactions on circuits and systems}, 28\penalty0
  (9):\penalty0 857--866, 1981.

\bibitem[Scieur et~al.(2017)Scieur, Roulet, Bach, and
  d'Aspremont]{scieur2017integration}
Scieur, D., Roulet, V., Bach, F., and d'Aspremont, A.
\newblock Integration methods and optimization algorithms.
\newblock \emph{Advances in Neural Information Processing Systems}, 30, 2017.

\bibitem[Surjanovic \& Bingham(2013)Surjanovic and Bingham]{testfunc}
Surjanovic, S. and Bingham, D.
\newblock Usage statistics of content languages for websites, 2013.
\newblock URL \url{https://www.sfu.ca/~ssurjano/optimization.html}.

\bibitem[Swenson et~al.(2021)Swenson, Murray, Poor, and
  Kar]{swenson2021distributed}
Swenson, B., Murray, R., Poor, H.~V., and Kar, S.
\newblock Distributed gradient flow: Nonsmoothness, nonconvexity, and saddle
  point evasion.
\newblock \emph{IEEE Transactions on Automatic Control}, 2021.

\bibitem[Wilson(2018)]{wilson2018lyapunov}
Wilson, A.
\newblock \emph{Lyapunov arguments in optimization}.
\newblock University of California, Berkeley, 2018.

\bibitem[Wilson et~al.(2021)Wilson, Recht, and Jordan]{wilson2021lyapunov}
Wilson, A.~C., Recht, B., and Jordan, M.~I.
\newblock A lyapunov analysis of accelerated methods in optimization.
\newblock \emph{Journal of Machine Learning Research}, 22\penalty0
  (113):\penalty0 1--34, 2021.

\bibitem[Wolfe(1969)]{wolfe1969convergence}
Wolfe, P.
\newblock Convergence conditions for ascent methods.
\newblock \emph{SIAM review}, 11\penalty0 (2):\penalty0 226--235, 1969.

\bibitem[Yuille \& Kosowsky(1994)Yuille and Kosowsky]{yuille1994statistical}
Yuille, A.~L. and Kosowsky, J.
\newblock Statistical physics algorithms that converge.
\newblock \emph{Neural computation}, 6\penalty0 (3):\penalty0 341--356, 1994.

\end{thebibliography}
\bibliographystyle{icml2023}

\newpage
\appendix
\onecolumn

\section{Proof of Theorem 1}
\label{main proof}

\begin{align}
    \frac{d}{dt}f(\xt)&=\langle \dfdxt,\dxdt(t)\rangle \\
    &= -\dfdxt^{\top} Z(\xt)^{-1}\dfdxt\\
    &\leq -d_1\Vert \dfdxt\Vert^2 \label{bounded dfdt}
\end{align}
Where \eqref{bounded dfdt} holds by Assumption \ref{z def}. 

Eq. \eqref{bounded dfdt} implies that the objective function is non-increasing along $\xt$.

For $t>0$ consider $\int_0^t \Vert\dfdxt\Vert^2dt$. By \eqref{bounded dfdt},
\begin{align}
    \int_0^t\Vert\dfdxt\Vert^2dt&\leq\frac{1}{d_1}(f(\x(0))-f(\x(t)))\\
    &\leq \frac{2R}{d_1} \label{bounded int norm}
\end{align}
Where \eqref{bounded int norm} holds by Assumption \ref{a1}.

Then for all $t>0$,
\begin{equation}
     \int_0^t\Vert\dfdxt\Vert^2dt\leq \frac{2R}{d_1}<\infty
\end{equation}

For all $\x$, $\y\in \R^n$,
\begin{align}
    \big \vert\Vert \dfdx \Vert^2 - \Vert \dfdy \Vert^2 \big \vert &\leq \big \vert \Vert \dfdx \Vert + \Vert \dfdy \Vert \big \vert \ \big \vert \Vert \dfdx \Vert -\Vert \dfdy \Vert \big \vert\\
    &\leq 2B \big \vert \Vert \dfdx \Vert - \Vert \dfdy \Vert \big \vert\\
    &\leq 2B \Vert \dfdx - \dfdy\Vert\\
    &\leq 2BL \Vert \x-\y \Vert
\end{align}

Therefore we can conclude that $\Vert \dfdx\Vert^2:\R^n\to\R$ is Lipschitz and hence uniformly continuous.

By Assumptions \ref{a4} and \ref{z def}, $\Vert \dxdt(t)\Vert$ is bounded and so $\xt$ is uniformly continuous in $t$.

Therefore, the composition $\Vert \dfdxt \Vert^2: \R_+\to\R$ is uniformly continuous in $t$. 

Since $\int_0^t\Vert\dfdxt\Vert^2dt<\infty$ and $\Vert\dfdxt\Vert^2$ is a uniformly continuous function of $t$, we conclude that $\lim_{t\to\infty}\Vert\dfdxt\Vert=0$.

\section{Boundedness of $Z^{-1}$}\label{zbounded}
\begin{lemma}
Assume \ref{a1}-\ref{a4} hold. Let $Z$ be defined as,
\begin{equation}
    Z_{ii}^{-1}(\xt) = \max\{\delta^{-1} [G(\xt) \nabla^2 f(\xt)\dfdxt]_i,1\}.
\end{equation}
For $\delta>0$, Assumption \ref{z def} will hold.
\end{lemma}

\textit{Proof:}
Note that the boundedness condition in Assumption \ref{z def} will hold if $0< \sum_{i=1}^n Z_{ii}(\xt)^{-1}<\infty$.

Lower bound:
\begin{align}
    Z_{ii}^{-1}(\xt) = \max\{\delta^{-1} [G(\xt) \nabla^2 f(\xt)\dfdxt]_i,1\}\geq 1>0.
\end{align}
Upper bound:
\begin{align}
     \sum_{i=1}^n Z_{ii}(\xt)^{-1} &\leq \mathbf{1}^{\top}\Big(\delta^{-1}G(\xt)\nabla^2f(\xt)\dfdxt\Big)+\sum_{i=1}^n 1\\
     &=\delta^{-1}\mathbf{1}^{\top}G(\xt)\nabla^2f(\xt)\dfdx +n\\
     &=\delta^{-1}\dfdxt^{\top}\nabla^2f(\xt)\dfdx +n\\
     &\leq \delta^{-1}\lambda_{\max}\big(\nabla^2 f(\xt)\big)\Vert \dfdxt\Vert^2+n\label{bound w lambda}\\
     &\leq \delta^{-1}\lambda_{\max}\big(\nabla^2 f(\xt)\big)B^2+n<\infty
\end{align}
Where $\lambda_{\max}\big(\nabla^2 f(\xt)\big)$ exists and is finite due to the existence of symmetric $\nabla^2 f$ \ref{a1}.

\section{Derivation of \eqref{z approx}}\label{z approx deriv}
Recall the approximation of $\frac{d}{dt}\dfdxt$ as in \eqref{ahat} and then use \eqref{time deriv} to find,
\begin{align}
    \ahat(\xt)&= \frac{\dfdxt - \nabla f(\x (t-\Delta t))}{\Delta t}\nonumber\\
    &\approx \frac{d}{dt}\dfdxt\\
    &=-\nabla^2 f(\xt) Z(\xt)^{-1} \dfdxt.
\end{align} 
Pre-multiply by $\delta^{-1}G(\x(t))$:
\begin{align}
    \delta^{-1}G(\xt)\ahat(\xt)\approx -\delta^{-1}G(\xt) \nabla^2 f(\xt) Z(\xt)^{-1} \dfdxt. \label{deriv 1}
\end{align}
As $Z$ is a diagonal matrix, the expression on the RHS of \eqref{deriv 1} can be rewritten as,
\begin{align}
    &-\delta^{-1}G(\xt) \nabla^2 f(\xt) Z(\xt)^{-1} \dfdxt\nonumber\\
    &=- \delta^{-1}G(\xt)\nabla^2 f(\xt)\begin{bmatrix}
    Z_{11}(\xt)^{-1}\frac{\partial f(\xt)}{\partial \x_1}\\ \vdots\\ Z_{nn}(\xt)^{-1}\frac{\partial f(\xt)}{\partial \x_n}
    \end{bmatrix} \label{deriv 2}\\
     &=-\delta^{-1}G(\xt)\nabla^2 f(\xt)\begin{bmatrix}
        \frac{\partial f(\xt)}{\partial x_1} & 0 &\dots & 0\\
        0  & \frac{\partial f(\xt)}{\partial x_2} & \ddots & 0\\
        \vdots & \vdots & \ddots & \vdots \\
        0 & 0 & \dots & \frac{\partial f(\xt)}{\partial x_n}
    \end{bmatrix}\begin{bmatrix}
        Z_{11}(\xt)^{-1}\\ \vdots \\ Z_{nn}(\xt)^{-1}
    \end{bmatrix}\\
    &=-\delta^{-1}G(\xt)\nabla^2 f(\xt)G(\xt)\z(\xt). \label{deriv 4}
\end{align}
In comparison, the objective is to compute:
\begin{align}
    \z(\xt)&= \delta^{-1}G(\xt) \nabla^2 f(\xt) \dfdxt \label{deriv 2}\\
    &= \delta^{-1}G(\xt) \nabla^2 f(\xt)G(\xt) \mathbf{1}\label{deriv 3}.
\end{align}
Where \eqref{deriv 2} follows from \eqref{final z} and \eqref{deriv 3} holds because by definition $G(\xt)_{ii} = [\dfdxt]_i$.

Define the shorthand notation for the expressions $\Ab\triangleq G(\xt)\nabla^2 f(\xt)G(\xt)$ and $\z\triangleq \z(\xt)$. Then \eqref{deriv 4} becomes $-\delta^{-1}\Ab\z$ and similarly \eqref{deriv 3} becomes $\delta^{-1}\Ab\mathbf{1}$. From \eqref{deriv 1} and \eqref{deriv 4},
\begin{equation}
    -\delta^{-1}\Ab \z \approx \delta^{-1}G(\xt)\ahat(\xt) \label{deriv 6}
\end{equation}
Use the definition of $\z$ in \eqref{deriv 3} and substitute into \eqref{deriv 6}:
\begin{align}
    \z &= \delta^{-1}\Ab \mathbf{1} \Rightarrow -\delta^{-1}\Ab \z = -\delta^{-2}\Ab\Ab\mathbf{1} \approx   \delta^{-1}G(\xt)\ahat(\xt)
    \end{align}
Pre-multiply by $-\mathbf{1}^{\top}$ and note that $\Ab = \Ab^{\top}$ :
    \begin{align}
    (\delta^{-1}\Ab\mathbf{1})^{\top}(\delta^{-1}\Ab\mathbf{1}) &\approx -\delta^{-1}\mathbf{1}^{\top}G(\xt)\ahat(\xt)\\
    \z^{\top}\z&\approx-\delta^{-1}\mathbf{1}^{\top}G(\xt)\ahat(\xt) \label{deriv 5}
\end{align}

One possible solution to \eqref{deriv 5} is thus,
\begin{align}
    &\z_i^2(\xt)\triangleq -[\delta^{-1}G(\xt)\ahat(\xt)]_i\\
    &\Rightarrow \z_i(\xt) = \sqrt{\max\{-[\delta^{-1}G(\xt)\ahat(\xt)]_i, 0\}}.
\end{align}

Similarly to the derivation of \eqref{z true}, we truncate $\z_i(\xt)\geq 1$ to finish the derivation.

\section{Boundedness of $\widehat{Z}^{-1}$} \label{zhat bounded}
\begin{lemma}
Assume \ref{a1}-\ref{a4} hold. Let $Z$ be defined as,
\begin{equation}
    \hat{Z}_{ii}(\xt)^{-1}=\max\{\sqrt{-\delta^{-1}[G(\xt)\ahat(\xt)]_i},1\}.
\end{equation}
For $\delta>0$, $0<\Delta t\leq t$, and finite $\x(0)$, Assumption \ref{z def} will hold.
\end{lemma}

\textit{Proof:} Note that the boundedness condition in Assumption \ref{z def} will hold if $0< \sum_{i=1}^n \hat{Z}_{ii}(\xt)^{-1}<\infty$.

Lower bound:
\begin{align}
    \hat{Z}_{ii}(\xt)^{-1}=\max\{\sqrt{-\delta^{-1}[G(\xt)\ahat(\xt)]_i},1\}\geq 1 >0.
\end{align}
Upper bound:
\begin{align}
    \hat{Z}_{ii}(\xt)^{-1} &\leq \max\{\sqrt{-\delta^{-1}[G(\xt)\ahat(\xt)]_i},1\}^2\\
    \sum_{i=1}^n \hat{Z}_{ii}(\xt)^{-1} &\leq  \sum_{i=1}^n -\delta^{-1}[G(\xt)\ahat(\xt)]_i+\sum_{i=1}^n 1\\
    &= \delta^{-1}(\Delta t)^{-1}\Big(\dfdxt^{\top} \nabla f(\x(t-\Delta t)) -\dfdxt^{\top}\dfdxt \Big)+n\\
    &=\delta^{-1}(\Delta t)^{-1}\dfdxt^{\top}\Big( \nabla f(\x(t-\Delta t)) -\dfdxt \Big)+n\\
    &\leq\delta^{-1}(\Delta t)^{-1}\left\vert\dfdxt^{\top}\Big(  \nabla f(\x(t-\Delta t)) -\dfdxt \Big)\right\vert+n\\
    &\leq \delta^{-1}(\Delta t)^{-1}\Vert \dfdxt\Vert \Vert \nabla f(\x(t-\Delta t)) -\dfdxt\Vert+n\\
    &\leq \delta^{-1}(\Delta t)^{-1} BL\Vert \x(t-\Delta t)-\x(t)\Vert+n\\
    &\leq \delta^{-1}(\Delta t)^{-1} BL\max_{\x^*\in S}\Vert \x(0)-\x^*\Vert+n<  \infty
\end{align}
Where the last line holds because $\Vert \x(0)-\x^*\Vert$ is bounded as $\x(0)$ is given to be finite and $\x^*$ is finite due to \ref{a2}.

\section{Computation Complexity}
\subsection{Full Hessian $Z$} \label{complexity full}
Given the gradient and Hessian, the per-iteration computation complexity to evaluate $Z$ according to \eqref{z true} is $\mathcal{O}(n^2)$. This can be verified as each element $i$ requires $\mathcal{O}(n)$ computations:
\begin{equation}
    Z_{ii}^{-1}(\xt) = \max\left\{\delta^{-1} \frac{\partial f(\xt)}{\partial \x_i(t)}\sum_{j=1}^n \frac{\partial^2 f(\xt)}{\partial \x_i \partial \x_j}\frac{\partial f(\xt)}{\partial \x_j}, 1\right\}
\end{equation}

\subsection{Approximate $\widehat{Z}$} \label{complexity approx}
Given the current and previous gradient, the per-iteration computation complexity to evaluate $\widehat{Z}$ according to \eqref{z approx} is $\mathcal{O}(n)$. This can be verified as each element $i$ requires $\mathcal{O}(1)$ computations:
\begin{equation}
    Z_{ii}^{-1}(\xt) = \max\left\{-\delta^{-1}\frac{1}{\Delta t} \frac{\partial f(\xt)}{\partial \x_i(t)} \left(\frac{\partial f(\xt)}{\partial \x_i}-\frac{\partial f(\x(t-\Delta t))}{\partial \x_i}\right), 1\right\}
\end{equation}

\section{More Details on Test Functions Experiment}\label{all toy graphs}

\subsection{All Graphs}
All the test function experiments can be found in Figure \ref{fig:test funcs}. The test functions considered were:
\begin{itemize}
    \item \textbf{Rosenbrock} \cite{testfunc} A convex, uni-modal, valley-shaped function.
    
    \item \textbf{Himmelblau} \cite{0070289212} A nonconvex, multi-modal, bowl-shaped function.

    \item \textbf{Booth} \cite{testfunc} A convex, uni-modal, plate-shaped function.

    \item \textbf{Three Hump} \cite{testfunc} A nonconvex, multi-modal, valley-shaped function.

    \item \textbf{Extended Wood} \cite{andrei2008unconstrained} A nonconvex function that scales for high dimensions.

    \item \textbf{Rastrigin} \cite{testfunc} A nonconvex, highly multi-modal function.
\end{itemize}

The compared methods were:
\begin{itemize}
    \item \textbf{Full Hessian ECCO} \eqref{z true}

    \item \textbf{Approximate ECCO} \eqref{z approx}

    \item \textbf{Gradient Descent} with step sizes chosen by line search and the Armijo condition

    \item \textbf{Gradient Descent} with step sizes chosen by EATSS, i.e. uncontrolled gradient flow discretized with FE+EATSS.

    \item \textbf{Adam} \cite{kingma2014adam} equipped with the full gradient 

\end{itemize}
Each experiment was terminated when convergence in the objective function was detected, i.e. when the algorithms returned $\x_k$ such that $\Vert f(\x_k)-f(\x_{k+1})\Vert <1e-4$. 

In all experiments, ECCO used $\delta=1$, $\alpha = 0.9$, $\beta=1.1$, $\eta=0.1$. The comparison methods used a grid search to optimize the hyperparameters. For Adam, we searched for $\beta_1$ and $\beta_2$ within $[0.7, 1]$ in increments of $0.01$. For both Adam and GD, we searched for an optimal learning rate within $[0.001,\dots, 1]$ in increments of $0.005$. The Armijo parameter was within $\{1e-5, 1e-4, 1e-3, 1e-2\}$.

These results empirically show that the approximation of the optimization trajectory does not induce much error into ECCO's performance, meaning that \eqref{z approx} is a reasonable substitution for \eqref{z true}. Second, FE+EATSS often outperforms FE with line search and the Armijo condition. Finally, Adam often performs badly on these test functions, which is expected as Adam excels as batch method in machine learning applications; we include it here for context.

\begin{figure}[h]
\centering
\subfigure[Rosenbrock $\x(0)=(-2,-2)$]{\label{fig:a}\includegraphics[width=0.3\textwidth]{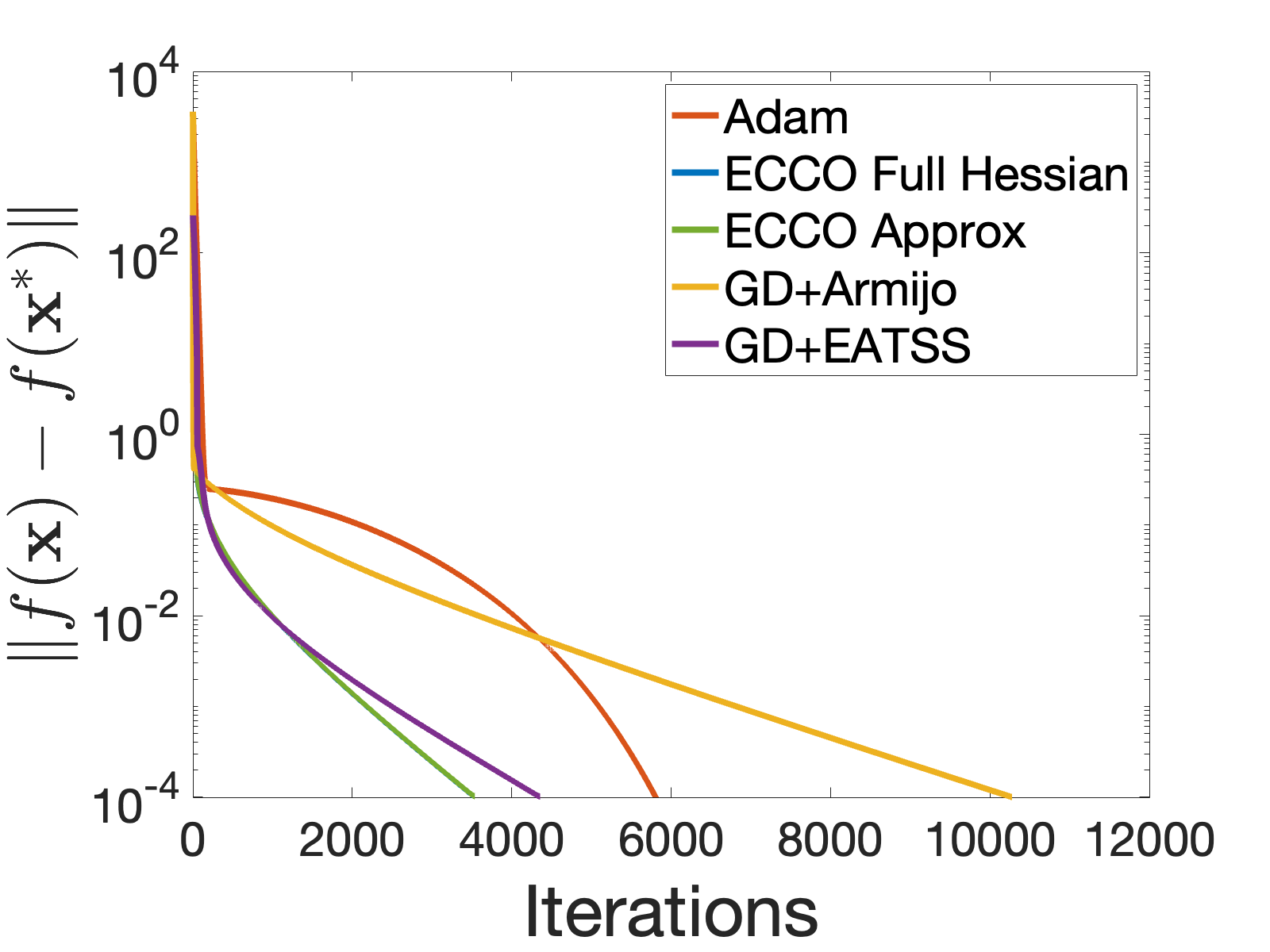}}
\subfigure[Rosenbrock $\x(0)=(0,0)$]{\label{fig:b}\includegraphics[width=0.3\textwidth]{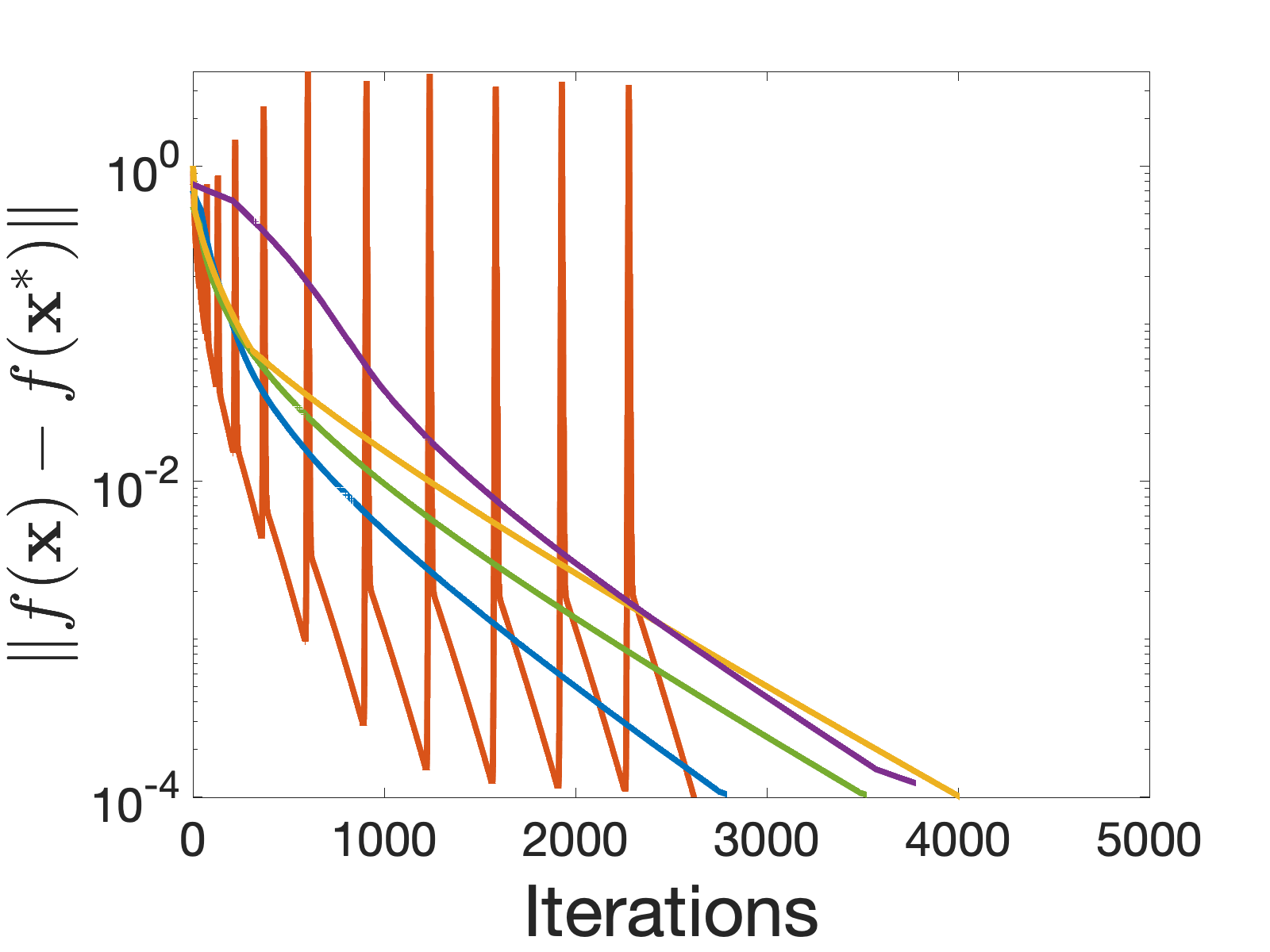}}
\subfigure[Rosenbrock $\x(0)=(-5,-5)$]{\label{fig:c}\includegraphics[width=0.3\textwidth]{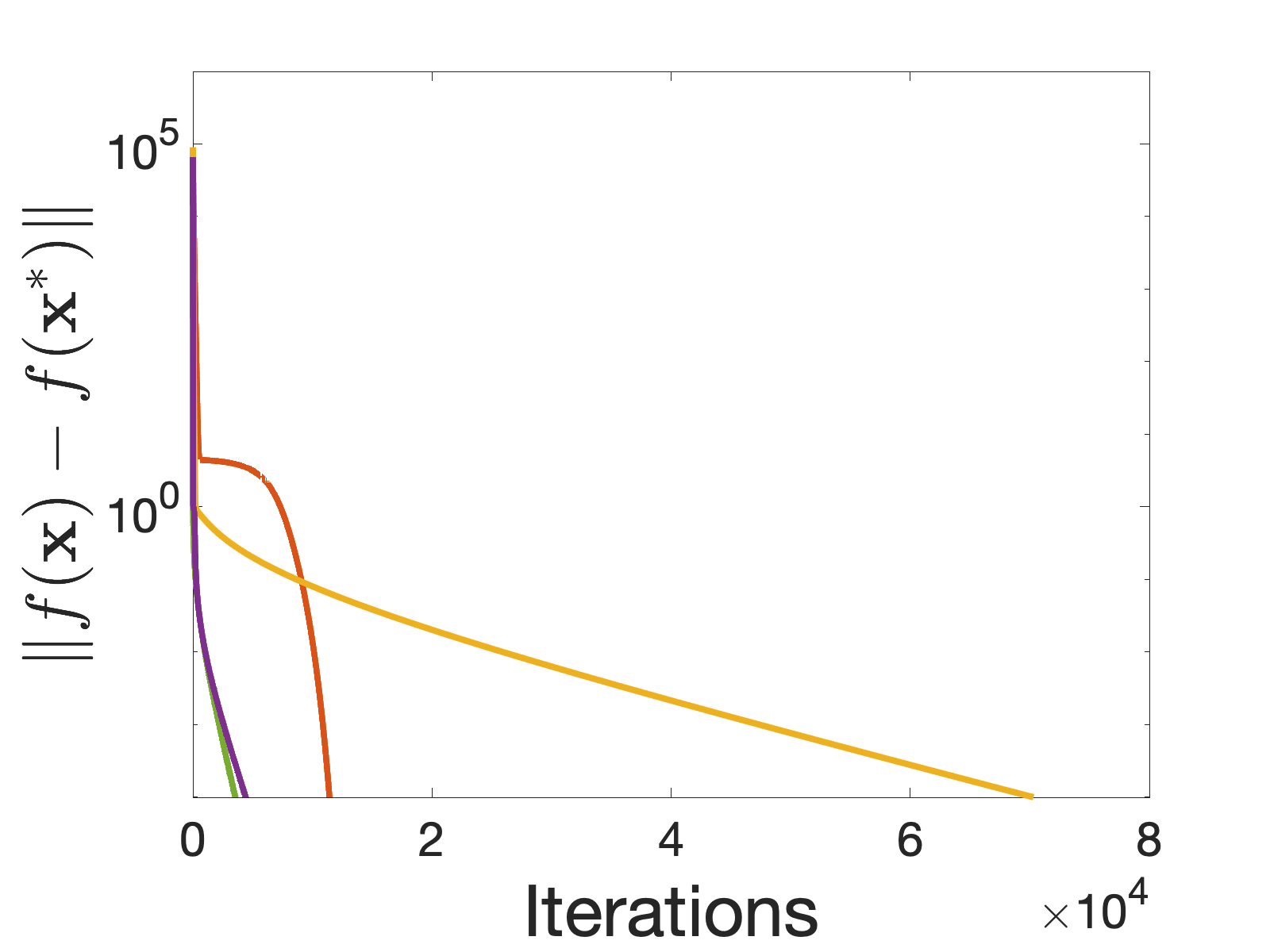}}

\subfigure[Himmelblau $\x(0)=(1,1)$]{\label{fig:d}\includegraphics[width=0.3\textwidth]{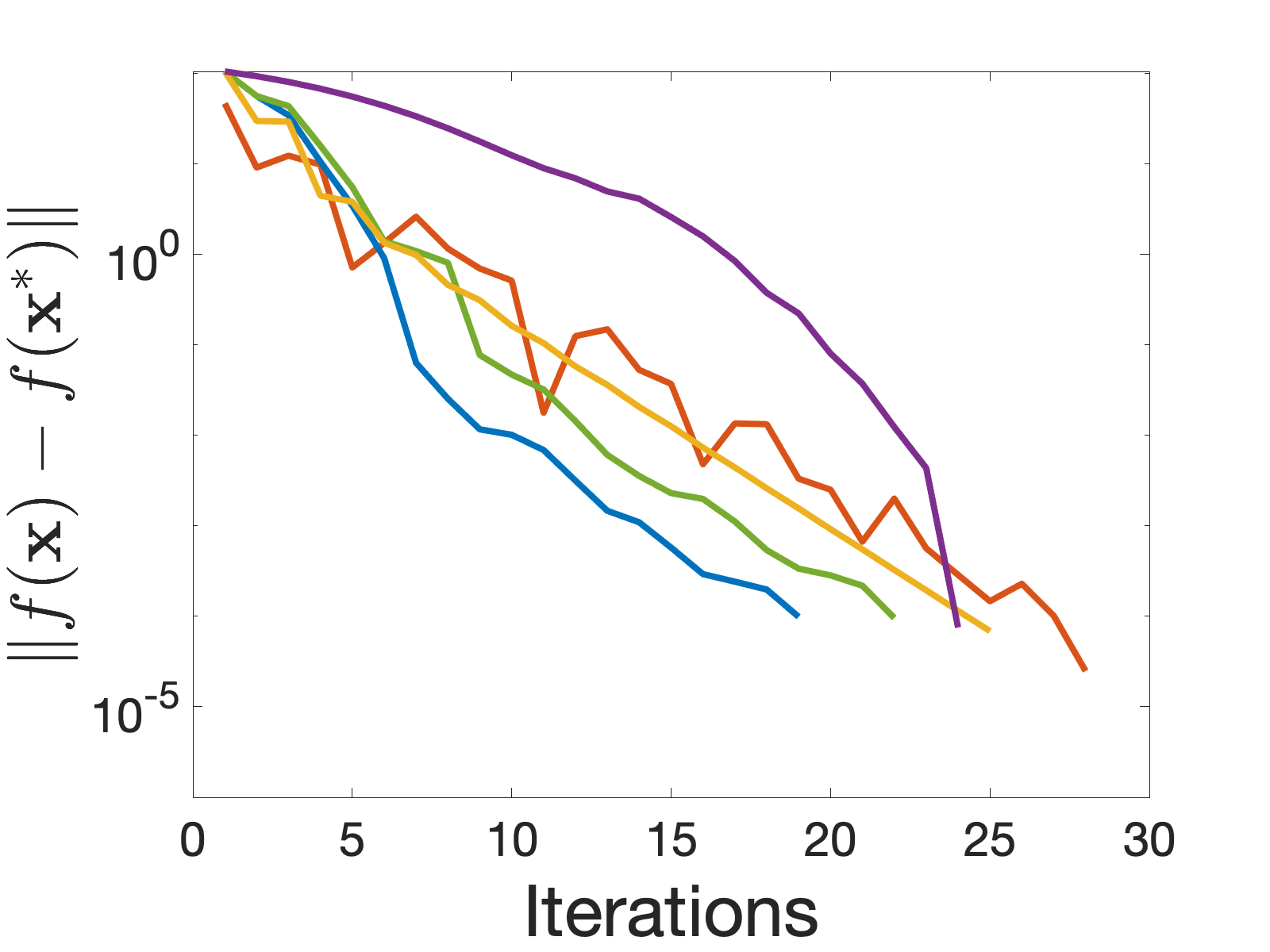}}
\subfigure[Himmelblau $\x(0)=(20,20)$]{\label{fig:e}\includegraphics[width=0.3\textwidth]{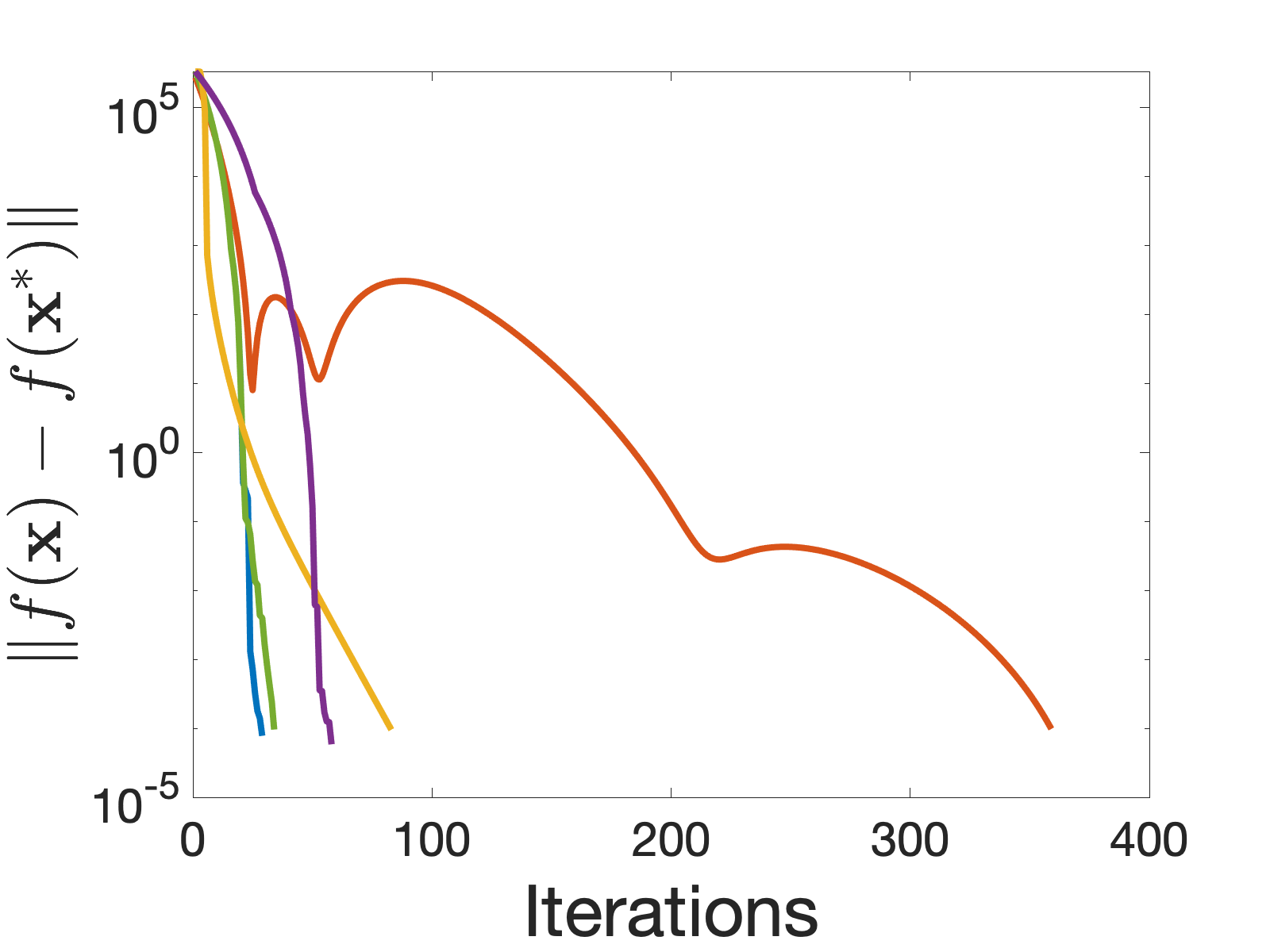}}
\subfigure[Himmelblau $\x(0)=(-5,-5)$]{\label{fig:f}\includegraphics[width=0.3\textwidth]{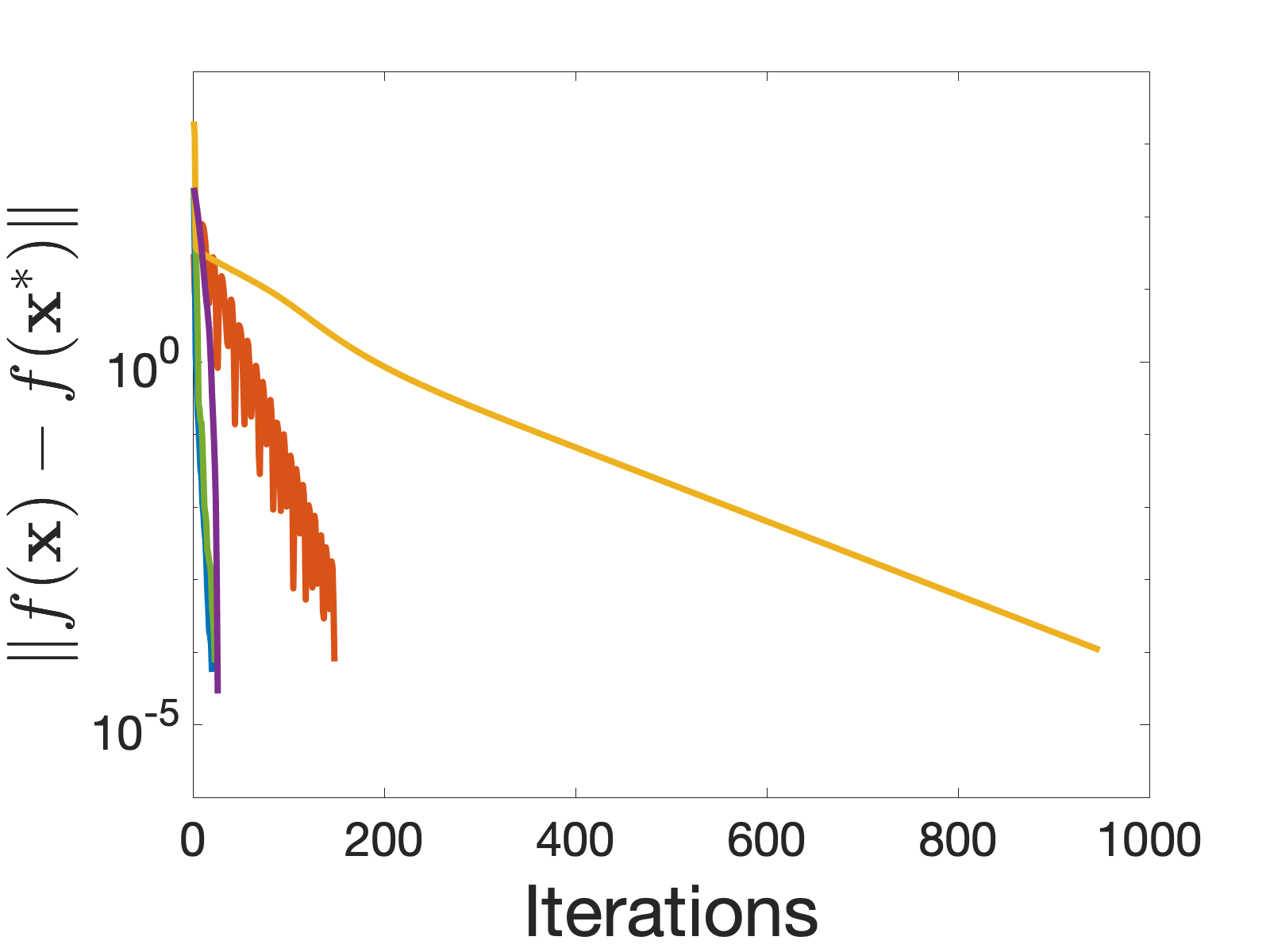}}

\subfigure[Booth $\x(0)=(5,5)$]{\label{fig:g}\includegraphics[width=0.3\textwidth]{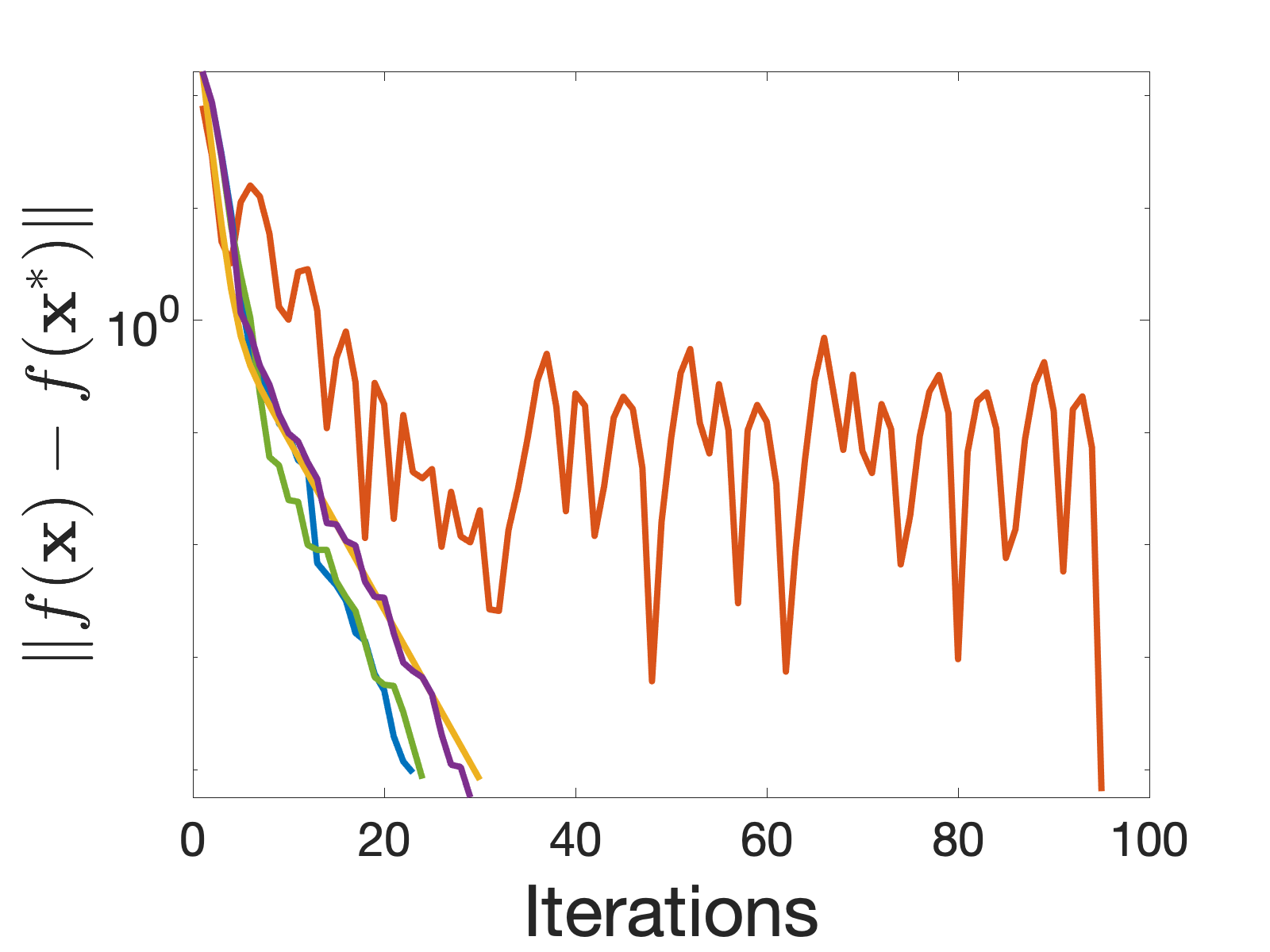}}
\subfigure[Booth $\x(0)=(5,-5)$]{\label{fig:h}\includegraphics[width=0.3\textwidth]{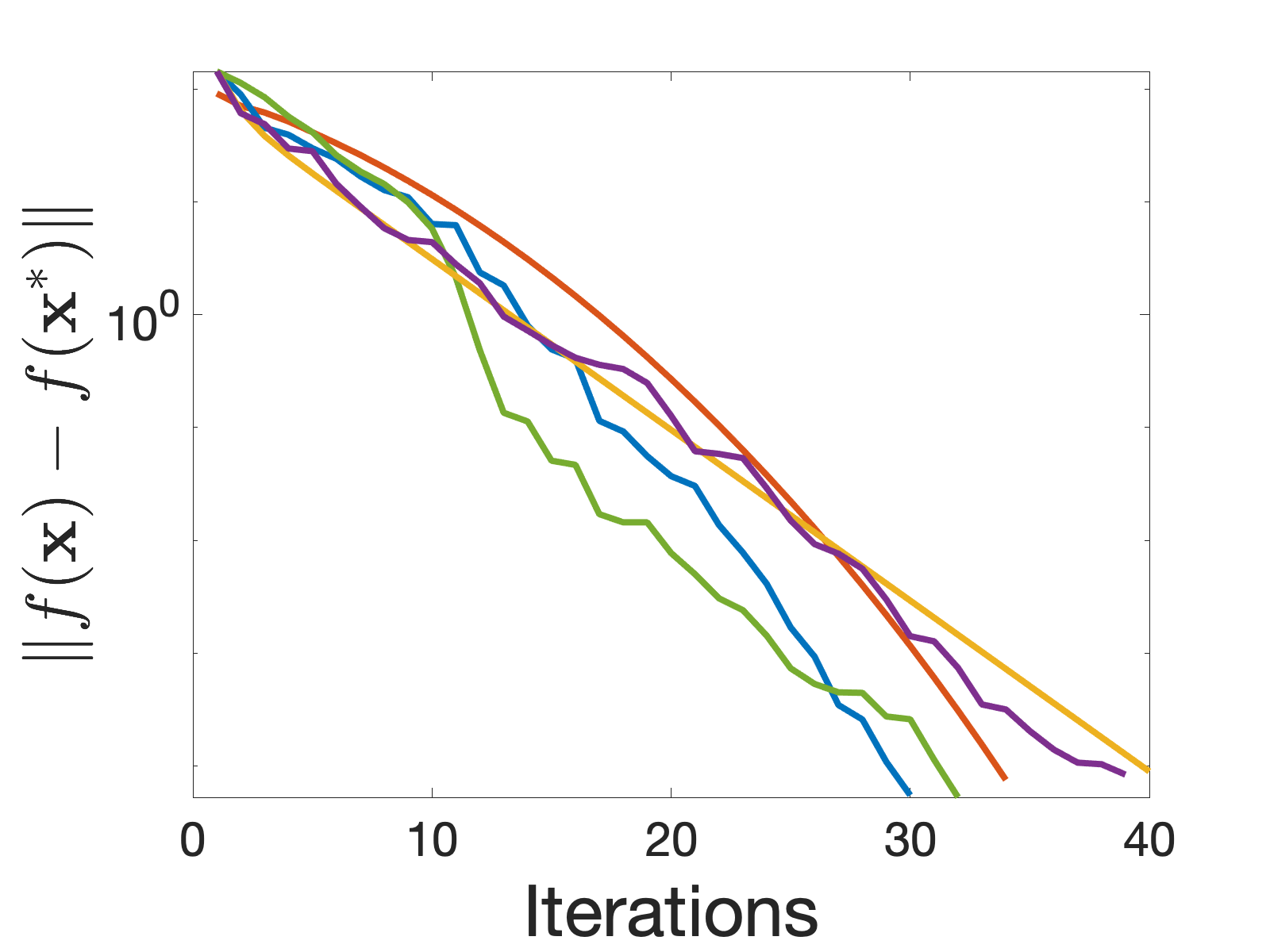}}
\subfigure[Booth $\x(0)=(-2,-2)$]{\label{fig:i}\includegraphics[width=0.3\textwidth]{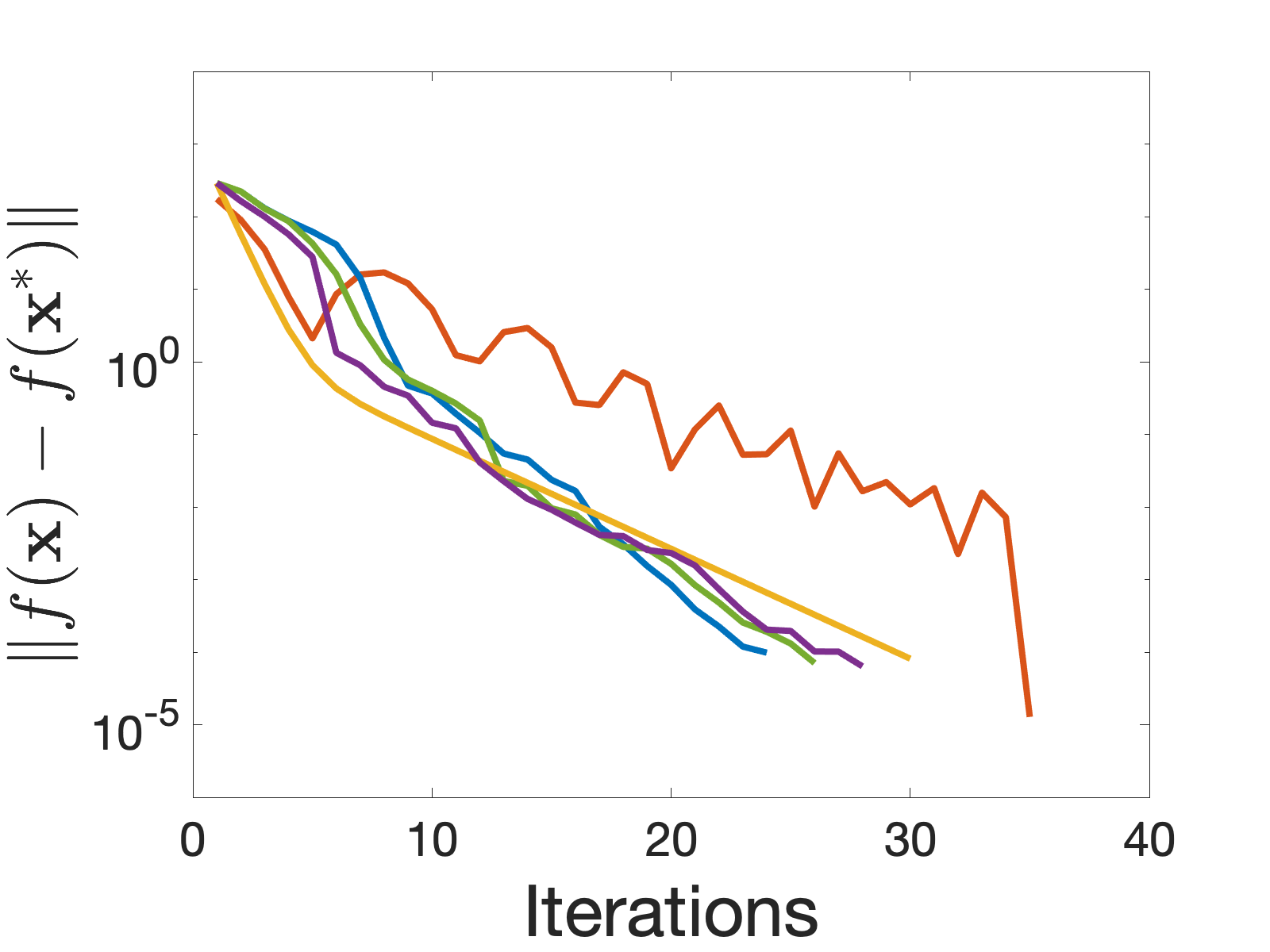}}

\subfigure[Three Hump $\x(0)=(1,1)$]{\label{fig:j}\includegraphics[width=0.3\textwidth]{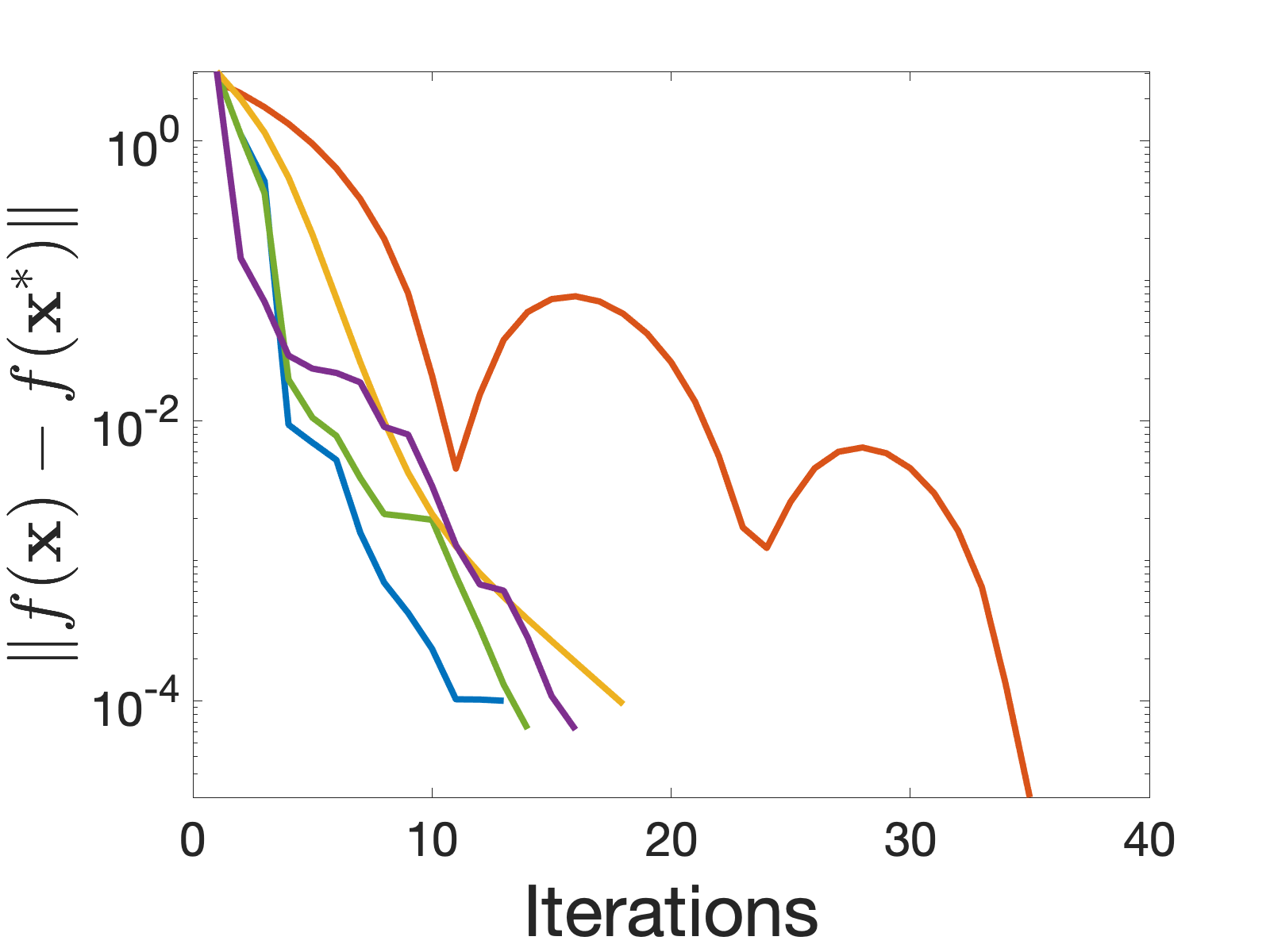}}
\subfigure[Three Hump $\x(0)=(0,-1)$]{\label{fig:k}\includegraphics[width=0.3\textwidth]{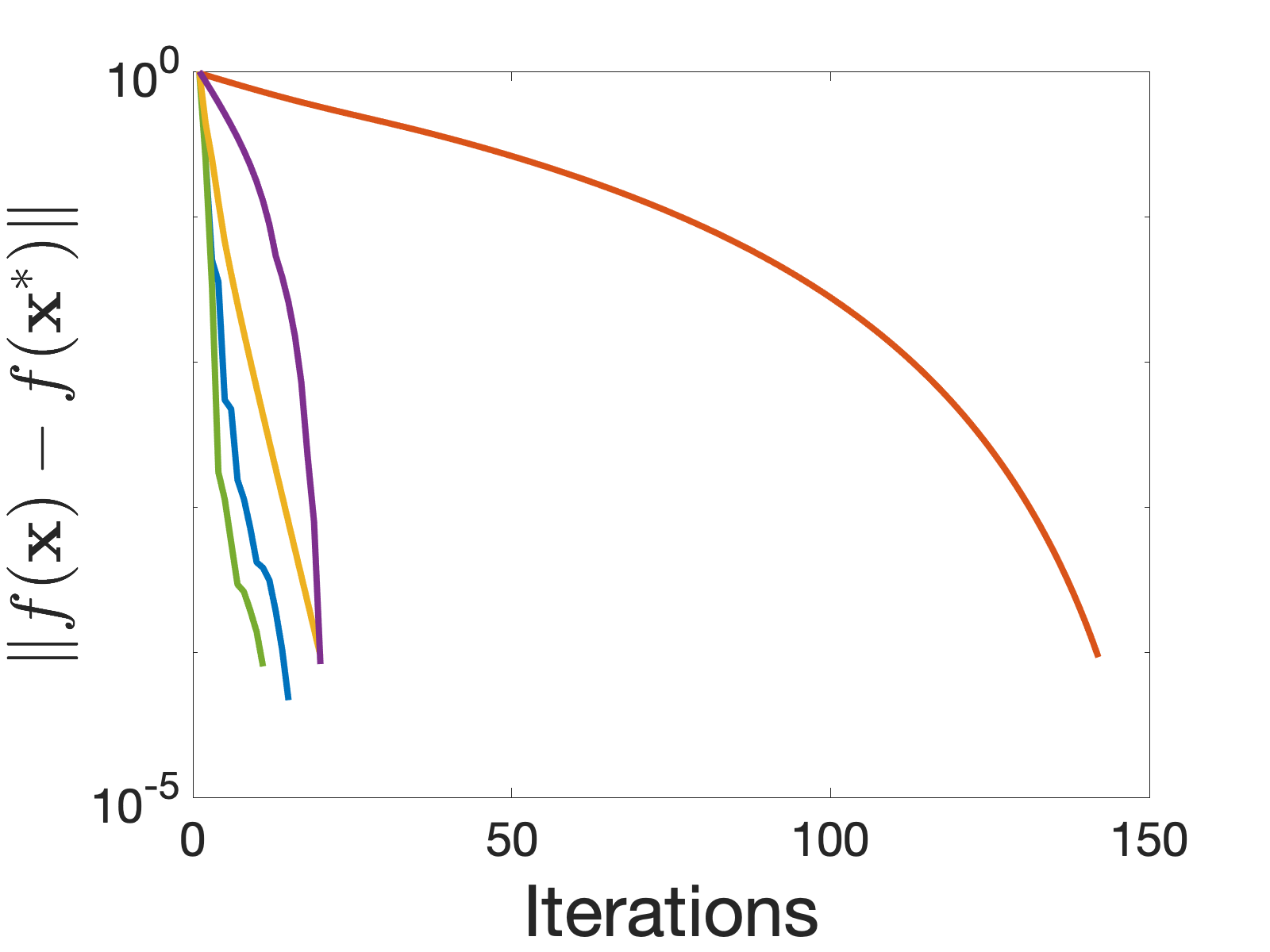}}
\subfigure[Three Hump $\x(0)=(-1,-1)$]{\label{fig:l}\includegraphics[width=0.3\textwidth]{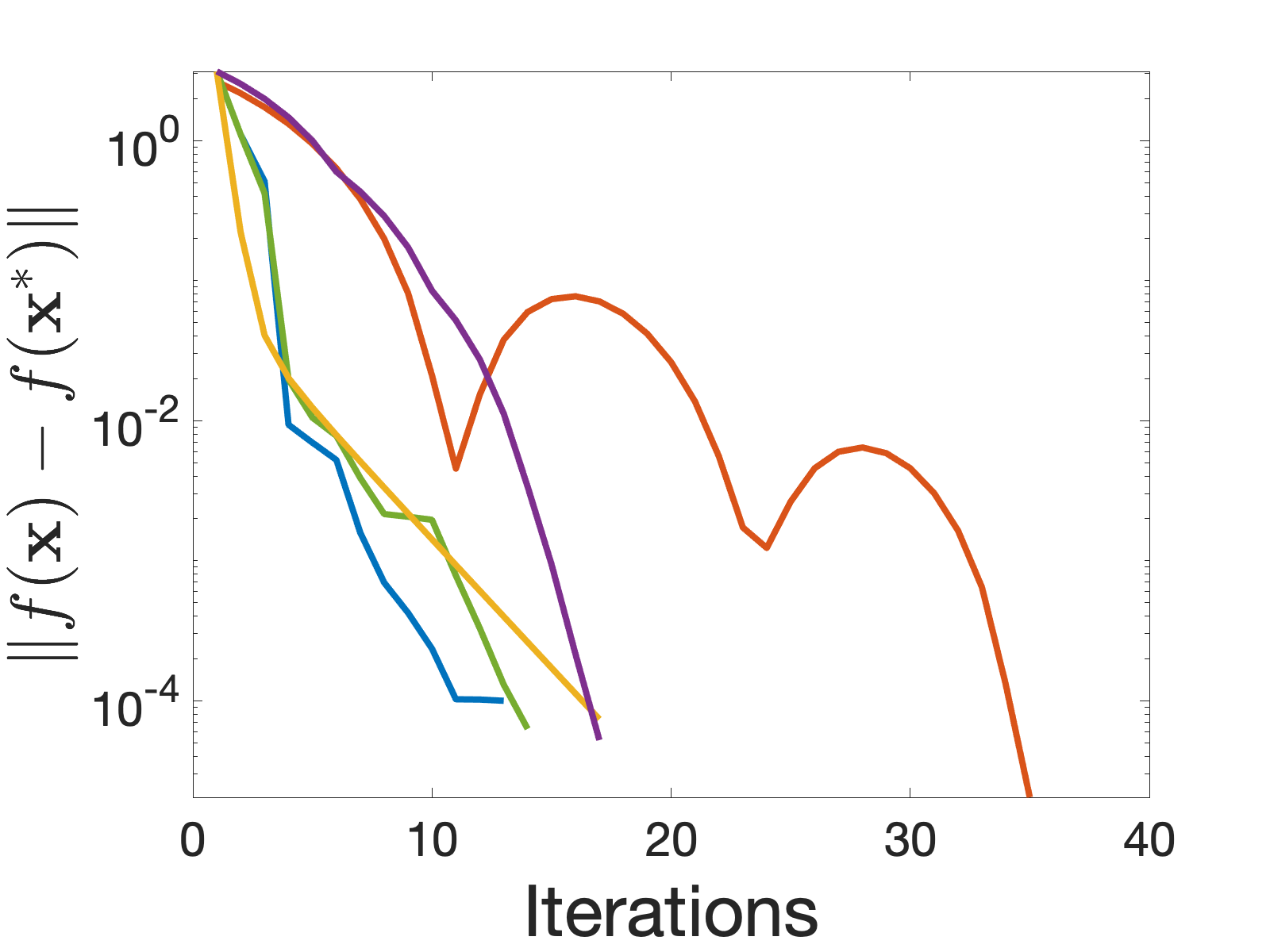}}

\subfigure[Rastrigin $\x(0)=(0.5, 0.5)$]{\label{fig:m}\includegraphics[width=0.3\textwidth]{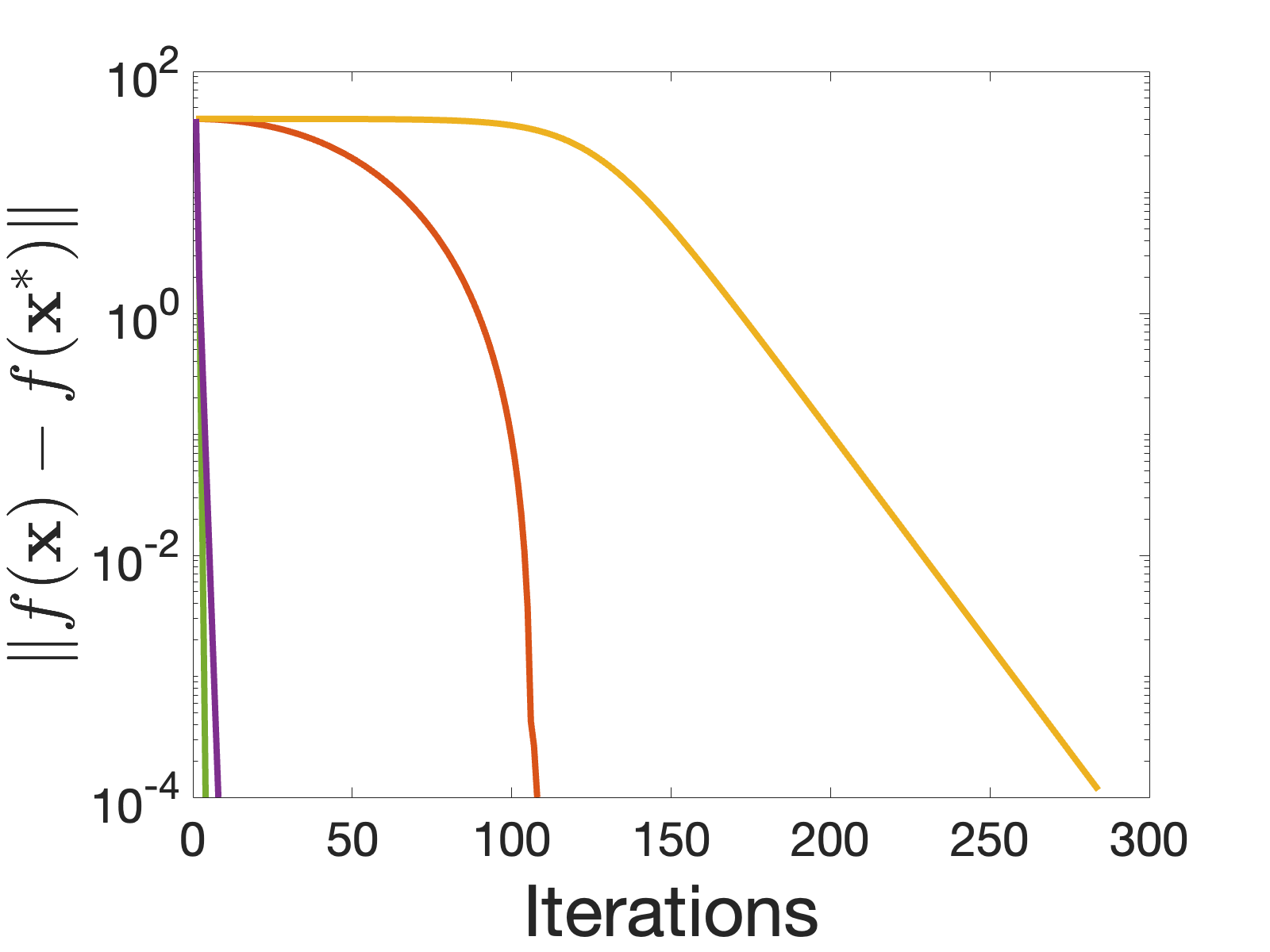}}
\subfigure[Extended Wood $\x(0)=2$]{\label{fig:n}\includegraphics[width=0.3\textwidth]{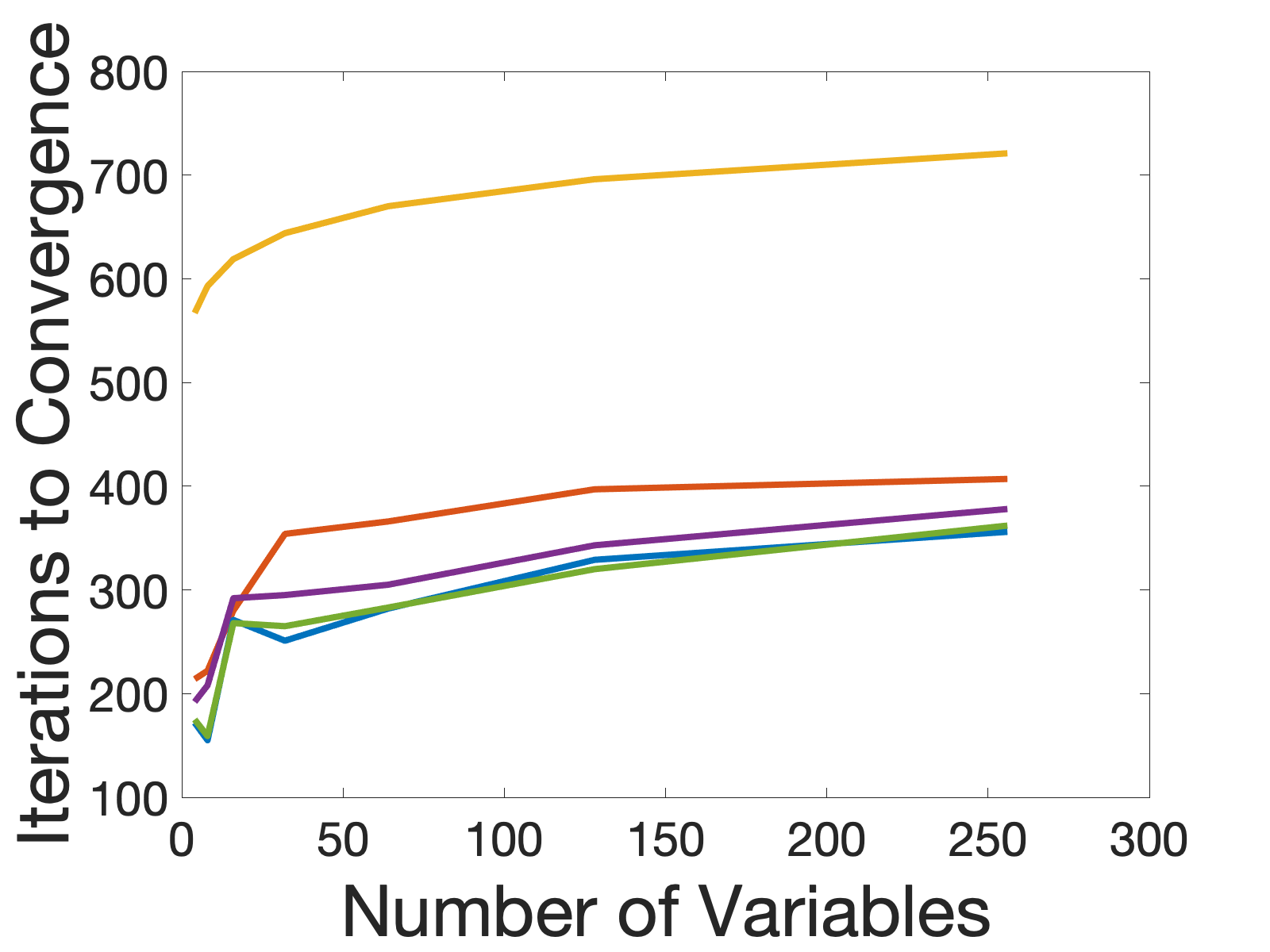}}
\subfigure[Extended Wood $\x(0)=10$]{\label{fig:o}\includegraphics[width=0.3\textwidth]{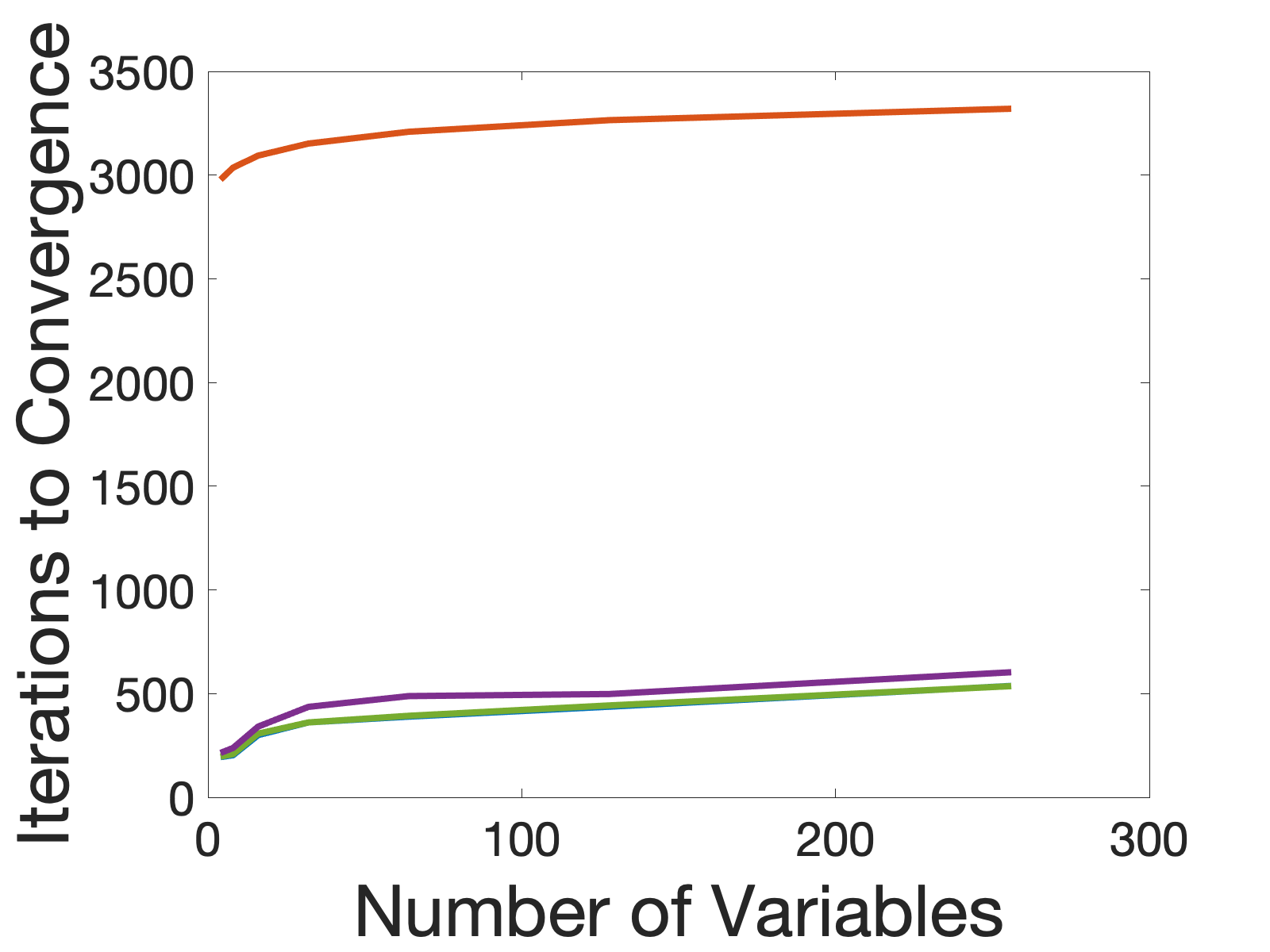}}

\caption{Test Functions}
\label{fig:test funcs}
\end{figure}

\subsection{Wall Clock Times}
See Table \ref{toy wall clock}.
\begin{table}[t]
\caption{\textbf{Test Functions Experiment}. Wall clock time comparison of tested optimization methods on test functions experiment, where each row is normalized to the ECCO full Hessian implementation. Note that this measures the wall clock time of each measure to convergence.}
\label{toy wall clock}
\vskip 0.15in
\begin{center}
\begin{small}
\begin{sc}
\begin{tabular}{lccccccr}
\toprule
Method  & ECCO \eqref{z true} & ECCO \eqref{z approx} &  GD+FE+Armijo & GD+FE+EATSS & Adam\\
\midrule
 Fig \ref{fig:a} & 1 & \textbf{0.982} & 1.68 & 3.15 &1.72\\
 Fig \ref{fig:b} & 1 & 1.24 & 1.31 & 1.31 & \textbf{0.91}  \\
 Fig \ref{fig:c} & 1 & \textbf{0.99} & 7.26 & 1.09 & 1.96  \\
 Fig \ref{fig:d} & \textbf{1} & 1.18 & 1.22 & 1.19 & 1.24  \\
 Fig \ref{fig:e} & 1 & \textbf{0.98} & 1.45  & 1.63 & 3.27 \\
 Fig \ref{fig:f} & 1 & 1.01 & 17.4 & \textbf{0.97} & 3.96  \\
 Fig \ref{fig:g} & \textbf{1} & 1.015 & 1.023 & 1.055 & 3.013 \\
 Fig \ref{fig:h} & \textbf{1} & 1.16 & 1.21 & 1.32  & 1.13 \\
 Fig \ref{fig:i} & \textbf{1} & 1.08 & 1.12 & 1.19 & 1.28  \\
 Fig \ref{fig:j} & \textbf{1} & 1.03 & 1.15 & 1.14 & 2.52 \\
 Fig \ref{fig:k} & 1 & \textbf{0.93} & 1.13 & 1.13 & 6.06  \\
 Fig \ref{fig:l} & \textbf{1} & 1.01 & 1.17 & 1.17 & 2.03  \\
 Fig \ref{fig:m} & 1 & \textbf{0.99} & 15.2 & 1.04 & 6.04  \\
 Fig \ref{fig:n} ($n=256$) & 1 & \textbf{0.92} & 1.16 & 2.44 & 1.15  \\
 Fig \ref{fig:o} ($n=256$) & 1 & \textbf{0.91} & DNC & 1.16 & 5.07  \\
\bottomrule
\end{tabular}
\end{sc}
\end{small}
\end{center}
\vskip -0.1in
\end{table}

\subsection{Backtracking Line Search Details} \label{backtrack}
Across the test functions, the initial time step \eqref{eq:lte} led to a sequence of time steps that satisfied the LTE and monotonicity checks in Algorithm \ref{time step algo} on the first or second try in 78\% of the iterations. This is to be expected as \eqref{eq:lte} is known from circuit theory to often satisfy the LTE criterion \cite{rohrer1981passivity}. We can thus observe that the time step search subroutine did not, in general, lead to high per-iteration complexity. 


\section{More Details on Neural Network Experiment} \label{more nn}
The neural network to classify the MNIST dataset (\cite{deng2012mnist}) used to train a neural network was modified from \cite{nnet_code}. The first layer has 50 nodes with a ReLU activation, the second had 20 nodes and a sigmoid activation, and the output layer had 10 nodes and a softmax activation.

\subsection{Hyperparameter Selection} \label{hyper search}
The search for finding optimal hyperparameters was accomplished by discretizing the parameter space of Adam, GD and RMSProp and performing grid search. For Adam, we searched for $\beta_1$ and $\beta_2$ within $[0.7, 1]$ in increments of $0.01$. For Adam, GD and RMSProp, we searched for an optimal learning rate within $[0.001,\dots, 1]$ in increments of $0.005$. For RMSProp, we searched for an optimal decay rate within $[0.1,\dots,1]$ in increments of $0.005$.

\subsection{Wall Clock Times}\label{nn wall clock}
See Table \ref{nn wall clock}. Note that computing \eqref{z approx} has the same per-iteration complexity as a gradient descent update. The experimental results show that the EATSS routine in Algorithm \ref{time step algo} did not substantially affect the wall clock time, i.e. by orders of magnitude.

\begin{table}[t]
\caption{\textbf{Neural Network Experiment}. Wall clock time comparison of tested optimization methods on neural network experiment, normalized to the ECCO full Hessian implementation.  Note that this is based on 200 iterations of training.}
\label{nn wall clock}
\vskip 0.15in
\begin{center}
\begin{small}
\begin{sc}
\begin{tabular}{lcccr}
\toprule
Method  & Mean Normalized Time Per Iteration \\
\midrule
ECCO \eqref{z approx}  &  1 \\
SGD                    &  0.93 \\
Adam                   &  0.96 \\
\bottomrule
\end{tabular}
\end{sc}
\end{small}
\end{center}
\vskip -0.1in
\end{table}

\subsection{More Robustness Experiments} \label{more robust}
We tested the robustness of the neural network to the hyperparameters of the selected optimization methods by perturbing the optimal hyperparameter values within a normalized ball of some radius and recording the accuracy of the trained neural network. The hyperparameter values were sampled from a uniform distribution as $\tilde\theta \sim U(\max\{{\theta}^*-\frac{\varepsilon}{{\theta}^*}, \underline{\theta}\},\min\{{\theta}^*+\frac{\varepsilon}{{\theta}^*},\bar{\theta} \}) $, where ${\theta}^*$ was the optimal value as found by grid search. Note that this perturbation was bounded to be within the domain $[\underline{\theta} ,\bar{\theta}]$ of the hyperparameter. 

In this section, we expand upon the results presented in the main paper by varying the radius for $\varepsilon\in\{0.01,1\}$ on ECCO and the comparison optimization solvers. For each fixed method and $\varepsilon$, 200 experiments were run, and the empirical classification accuracies are reported in Figures \ref{fig:supp_class_acc_vareps_0.01} and \ref{fig:supp_class_acc_vareps_1}. Note that the neural networks trained with ECCO are highly robust to the perturbations in the parameters; we found that when the neural network was trained with ECCO, we could perturb the optimal hyperparameters within a ball of $\varepsilon=1$ and lose at most 8\% in accuracy. 

ECCO's results are dramatically better than any of the comparison methods, where the accuracy suffered with small perturbations. Adam trained well for $\varepsilon=0.01$, but lost consistency and thus reliability for $\varepsilon=0.1$ and $\varepsilon=1$. Gradient descent and RMSProp similarly did not display robustness to the perturbations, implying that their methods need careful parameter tuning to be successful. 

\begin{figure}
    \centering
    \includegraphics[width=0.6\linewidth]{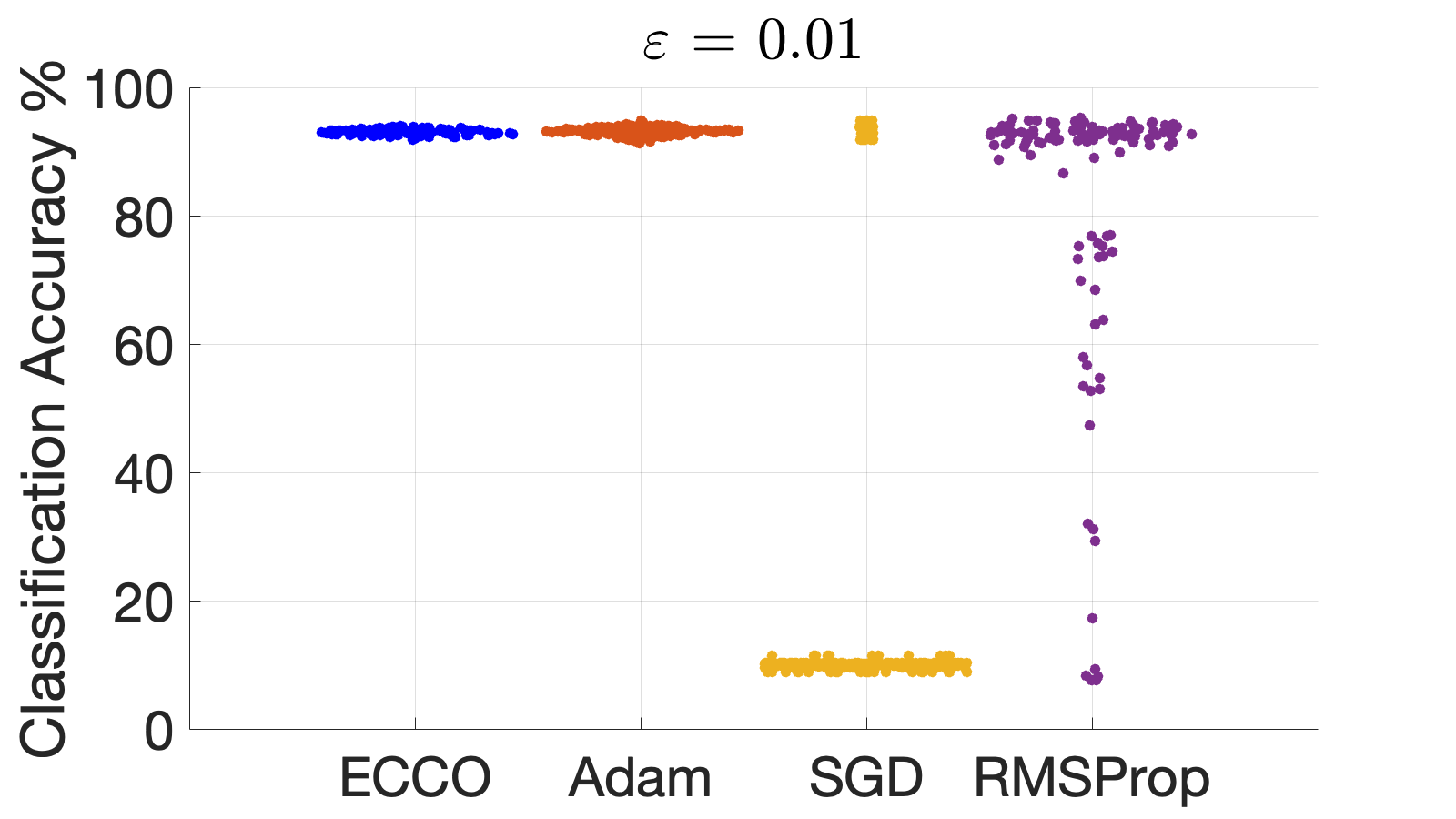}
    \caption{\small{Classification accuracy of 3-layer neural network using ECCO, Adam, gradient descent, and RMSProp with a random sampling of hyperparameters within a ball of $\varepsilon=0.01$ around the optimal hyper parameter values. Note that ECCO and Adam successfully train the neural network, while gradient descent and RMSProp are not able to consistently train in this regime.}}
    \label{fig:supp_class_acc_vareps_0.01}
\end{figure}

\begin{figure}
    \centering
    \includegraphics[width=0.6\linewidth]{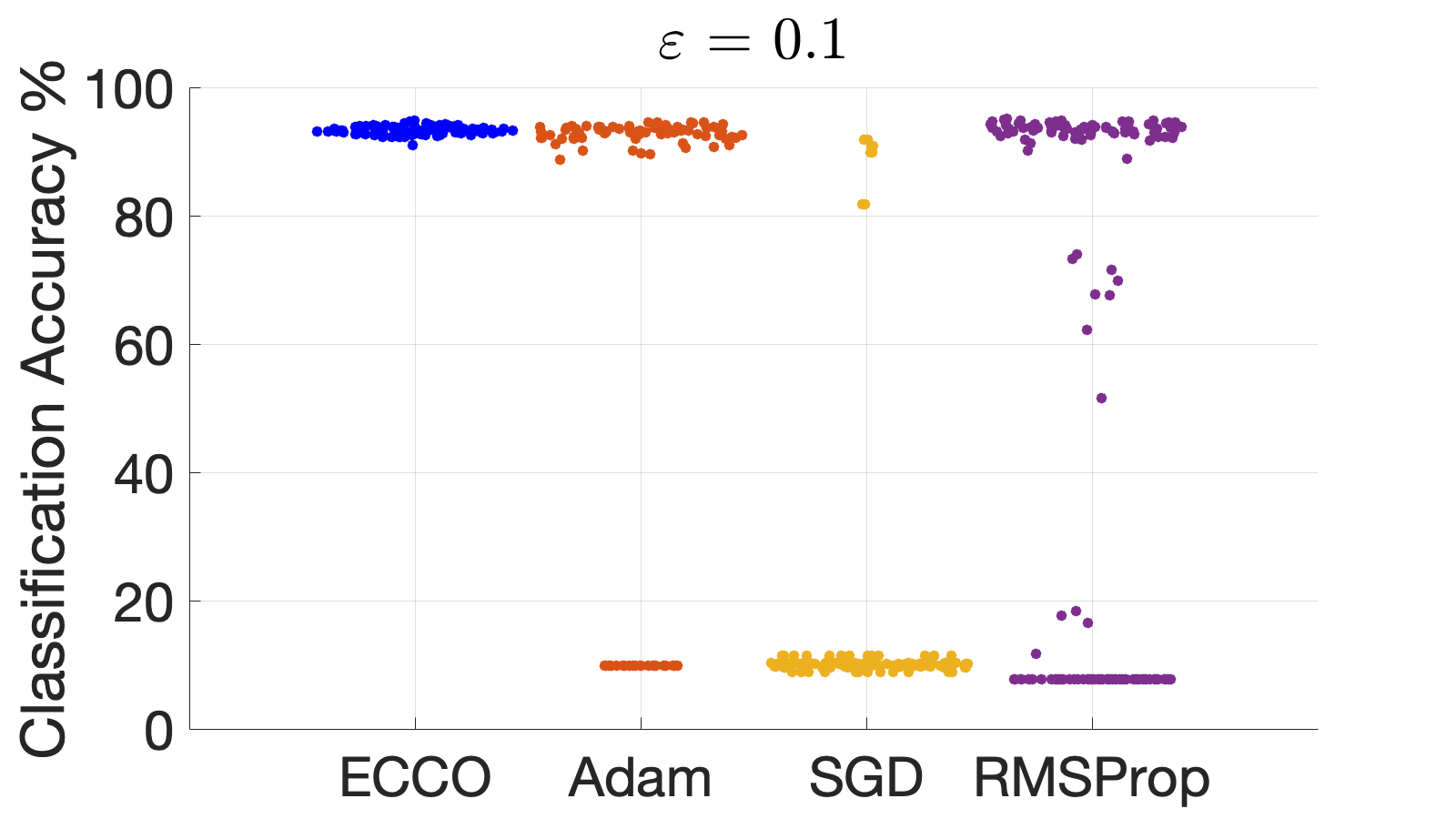}
    \caption{\small{Classification Accuracy of 3-layer neural network using ECCO, Adam, gradient descent, and RMSProp with a random sampling of hyperparameters within a ball of $\varepsilon=0.1$ around the optimal hyper parameter values. Note that ECCO is able to successfully train the neural network; however, none of the comparison methods are able to reliably train in this regime.}}
    \label{fig:supp_class_acc_vareps_1}
\end{figure}

In addition, perturbations to the comparison methods easily caused divergence. The percentage of experiments for which the optimization methods successfully converged to a fixed point is reported in Figure \ref{diverged table}. Note that ECCO was always able to converge due to the choice of discretization, but perturbations to the optimally tuned hyperparameters for Adam, gradient descent, and RMSProp easily caused those methods to diverge. If divergence occurred, the neural network was not trained, and it yielded an accuracy of about 10\% which was equivalent to random guessing. This behavior is easily visible by observing Figures \ref{fig:robust in paper}, \ref{fig:supp_class_acc_vareps_0.01}, and \ref{fig:supp_class_acc_vareps_1}.

\begin{table}[t]
\caption{\textbf{Convergence of Experiments.} Percent of experiments with perturbed hyperparameters for which the optimization methods converged to an approximate fixed point. Note that ECCO was always able to converge, but perturbing the comparison methods easily caused divergence.}
\label{diverged table}
\vskip 0.15in
\begin{center}
\begin{small}
\begin{sc}
\begin{tabular}{lcccccr}
\toprule
 & ECCO Approx \eqref{z approx} & Adam & SGD & RMSProp \\
\midrule
$\varepsilon = 0.01$   & 100\% & 100\% & 61.01\% & 80.78\%  \\
$\varepsilon = 0.1$    & 100\% & 77.2\% & 26.10\% & 61.01\%  \\
$\varepsilon = 1$     & 100\% &  65.4\%& 2.51\% & 53.09\%  \\
\bottomrule
\end{tabular}
\end{sc}
\end{small}
\end{center}
\vskip -0.1in
\end{table}

These experiments suggest that ECCO may be used for neural network training for similar performance to state-of-the-art methods like Adam, but without the extensive need for hyperparameter tuning. This suggests ECCO would need less data than Adam for cross-validation procedures and is apt for generalization and distribution shift.




\end{document}